\documentclass[letterpaper]{article}
\pdfoutput=1
\usepackage{hyperref}

\usepackage{amsmath,amsthm,amssymb,enumerate}
\usepackage{algorithm,algorithmic}
\usepackage{fancyhdr,dsfont,graphicx}
\usepackage{caption}
\usepackage{subcaption}
\DeclareRobustCommand{\wron}[4]{W^{#1,#2}_{#3,#4}}
\newtheorem{prop}{Proposition}

\newtheorem{lemma}{Lemma}

\def\orho{\alpha}
\def\P{\mathbb{P}}
\def\E{\mathds{E}}

\def\R{\mathds{R}}
\def\C{\mathds{C}}
\def\s{\mathfrak{s}}
\def\m{\mathfrak{m}}
\def\I{\mathcal{I}}
\def\M{\mathcal{M}}
\def\U{\mathcal{U}}
\def\B{B}
\def\H{\mathsf{H}}
\def\h{\mathsf{h}}
\def\e{\mathrm{e}}
\def\d{\mathrm{d}}
\def\G{\mathcal{G}}

\def\F{\mathsf{F}}
\def\X{\mathsf{X}}

\def\indfunc{\mathds{1}}
\def\deq{\mathrel{\mathop:}=}

\title{Pricing Step Options under the CEV and other Solvable Diffusion Models}

\author{G. Campolieti, R.~Makarov, and K. Wouterloot\\
Department of Mathematics, Wilfrid Laurier University\\
Waterloo, Ontario, Canada\\ \texttt{gcampoli@wlu.ca}, \texttt{rmakarov@wlu.ca} and \texttt{kwouterloot@gmail.com}}

\date{January 25, 2013}

\begin{document}

\maketitle

\begin{abstract}
We consider a~special family of occupation-time derivatives, namely proportional step options introduced by
Linetsky in [Math. Finance, 9, 55--96 (1999)]. We develop new closed-form spectral expansions for pricing such
options under a~class of nonlinear volatility diffusion processes which includes the constant-elasticity-of-variance (CEV) model
as an example. In particular, we derive a~general analytically exact expression for the resolvent kernel (i.e. Green's function) of such
processes with killing at an exponential stopping time (independent of the process) of occupation above or below a~fixed level. Moreover, we
succeed in Laplace inverting the resolvent kernel and thereby derive newly closed-form spectral expansion formulae for the transition probability
density of such processes with killing. The spectral expansion formulae
are rapidly convergent and easy-to-implement as they are based simply on knowledge of a~pair of fundamental solutions for an underlying solvable diffusion process.
We apply the spectral expansion formulae to the pricing of proportional step options for four specific families of solvable nonlinear diffusion asset price
models that include the CEV diffusion model and three other multi-parameter state-dependent local volatility confluent hypergeometric diffusion processes.
\end{abstract}

\section{Introduction} \label{sec:1}
Consider a~ continuous-time stochastic asset (e.g. stock) price process $\mathbf{S}=\{ S_t \}_{t\geq 0}$.
We recall that the occupation times $A_T ^{L,\pm} \equiv A_{T,\mathbf{S}}^{L,\pm}$ of the process $\mathbf{S}$
below a~given level $L>0$, during the time interval $[0,T]$, is defined by
\begin{equation*}
A_T ^{L,-} \deq \int_0 ^T \indfunc_{S_t \leq L}\,\d t
\end{equation*}
and, in a~similar way, the occupation time for staying above level $L$ is
\begin{equation*}
A_T ^{L,+} \deq T - A_T ^{L,-} = \int_0 ^T \indfunc_{S_t > L}\,\d t .
\end{equation*}

Occupation time derivatives were introduced as a~more flexible alternative to standard barrier and lookback options. Many types of occupation time options have been proposed such as the step, $\alpha$-quantile, Parisian, and corridor options (e.g. see \cite{Dass95,FT01,Linet99}).  A family of proportional step options was introduced in \cite{Linet99} as a~flexible alternative to knock-out barrier options. Step options do not lose their value when the barrier is reached. Instead, the payoff function of such options depends continuously on the time that the underlying price spends below or above a~ given (barrier) level. For a~proportional step option, its payoff at maturity is defined as the payoff of a~vanilla option discounted by a~factor
$\exp \left( -\orho A_T \right)$, where $A_T$ is an occupation time. For maturity time $T>0$ and given level $L>0$, the payoff functions of the proportional "down-and-out" ($f_\text{step}^-$) and "up-and-out" ($f_\text{step}^+$) step options for the asset price process $\{ S_t \}_{t\geq 0}$ take the form:
\begin{equation}
\label{StepPayoff}
  f_\text{step}^\pm  \equiv \exp \left( -\orho A_{T,\mathbf{S}}^{L,\pm} \right) f(S_T).
\end{equation}
Here, $\orho > 0$ is a~parameter and $f(S_T)$ is a~payoff of a~vanilla option, e.g. $(S_T - K)^+$ for a~European call or $(K - S_T)^+$ for a~put for
a given strike $K>0$. We note that for $\orho = 0$ the proportional step option is simply a~vanilla European option with payoff $f(S_T)$. For any occupation time
$A_T \ge 0$, the factor $\e^{-\orho A_T}$ is non-increasing in $\orho$. Recall the maximum $M_T = \sup\{S_t: 0\le t \le T\}$ and minimum
$m_T = \inf\{S_t: 0\le t \le T\}$ of the process up to time $T$. By continuity of the process $\mathbf{S}$ on $\R_+$, for $S_0 < L$ we have
$\e^{-\orho A_T^{L,+}} \ge \indfunc_{M_T < L}$ and $\e^{-\orho A_T^{L,+}} \to \indfunc_{M_T < L}$ a.s., as $\orho\to\infty$. Similarly,
for $S_0 > L$, $\e^{-\orho A_T^{L,-}} \ge \indfunc_{m_T > L}$ and $\e^{-\orho A_T^{L,-}} \to \indfunc_{m_T > L}$ a.s., as $\orho\to\infty$.
We hence observe that proportional step options may be regarded as an interesting alternative to knock-out barrier options with less risky payoffs but
approach their knock-out barrier counterparts as the discount penalty factor $\orho\to\infty$.

Closed form pricing and hedging formulae have been found for many types of occupation time derivatives under the assumption that the underlying asset price follows a~geometric Brownian motion (e.g. see \cite{Dass95,FT01,Hugo99,Linet99}).  However, analytical pricing of occupation time options is a~non-trivial problem for nonlinear diffusion asset price models. Recently, a~Laplace transform-based approach to price occupation time derivatives under Kou's double exponential jump diffusion model was presented in \cite{CCW10}. In \cite{LK07}, the double Laplace transform of the joint probability density function (PDF) of the asset value and occupation time was obtained for the CEV diffusion model.

In this paper, we present an analytical method for pricing proportional step options under a~class of solvable diffusion models.  The models considered
include the constant elasticity of variance (CEV) diffusion model and other confluent hypergeometric processes. Our approach uses a~closed-form spectral
expansion for the transition density of the process with killing at an exponential stopping time of occupation for the process above or below a~fixed level.
This is essentially the Feynman-Kac theorem combined with an analytical inversion method for the Laplace transform of the Green's function. The Feynman-Kac
theorem has been used in the past to find an expression for joint probability distribution of $S_T$ and $A_T^{L,\pm}$, assuming that $\{ S_t \}_{t\geq 0}$
follows a~Brownian motion. It has also been used to generate analytical prices for occupation time options assuming that the underlying asset follows a~GBM
in \cite{Hugo99,Linet99}, as well as Kou's model in \cite{CCW10}.
In \cite{LK07}, the Feynman-Kac approach is proposed for deriving the double Laplace transform of the joint probability density function (PDF) of $S_T$ and $A_T^{L,\pm}$,
assuming that the underlying asset follows a~CEV process.  In principle, this method can be used to price proportional step options and any other
options whose payoff functions depend only on $S_T$ and $A_T^{L,\pm}$, and it can easily be generalized to a~class of solvable diffusions.
%
In this paper, we follow the approach similar to that of \cite{LK07}. In contrast to \cite{LK07}, we are able to derive computationally tractable pricing formulae for proportional step options. By applying the Feynman-Kac formula, we obtain the Laplace transform of the transition PDF of the asset price process with killing at an exponential stopping time of occupation for the process above or below a~fixed level.
The residue theorem allows us to invert the Laplace transform. The resulting pricing formula is given in the form of an integral of a~spectral series expansion.

Our pricing method for proportional step options does not rely on computing the joint PDF of $S_T$ and $A_T^{L,\pm}$.
Although, such a joint PDF is also obtainable by a single Laplace inversion of the above mentioned transition PDF.
Moreover, the approach in this paper can be readily generalized
to other step options with a~different structure of the occupation time (e.g. corridor options) as the the relevant transition densities follow by the same spectral expansion methodology presented in this paper.
For example, it can be applied to the case of a~double-barrier proportional step option where the payoff depends on the occupation time of the underlying process in between two barriers $L_1$ and $L_2$ with $L_1<L_2$ (see \cite{Linet02}). Such an occupation time is given by
\begin{equation*}
A_{T,\mathbf{S}}^{L_1,L_2} = \int_0 ^T \indfunc_{L_1 <S_t < L_2}\,\d t.
\end{equation*}
We also note that the spectral expansion approach in this paper is generally applicable to pricing step options in the presence of knock-out barriers for the solvable models considered.

The organization of our paper is as follows. In Section~\ref{sec:2}, we present the general framework for the no-arbitrage evaluation of step options and their deltas under solvable diffusion models. This section also contains our main result on the Green's functions and the corresponding closed-form spectral expansions for the transition probability densities with killing at an exponential stopping time of occupation above or below a~fixed level. In Section~\ref{sec:4}, we give explicit analytical expressions for the CEV model \cite{Cox75}
and the other solvable nonlinear local volatility models \cite{CM10}.
Section~\ref{sec:5} presents the methodology for computing the step options.
In particular, the actual implementation of the spectral expansions is given in
Subsection~\ref{subsec:5.1} and a~Monte Carlo bridge sampling approximation approach is presented in Subsection~\ref{subsec:5.2}.
Numerical results for pricing and hedging proportional step options under the CEV model and other classes of confluent hypergeometric diffusion models are presented in Section~\ref{subsec:5.3} and some conclusions are drawn in Section~\ref{sec:7}.
Supporting proofs of our main theoretical results are contained in the Appendix.

\section{Analytical Pricing of Step Options} \label{sec:2}
Consider a~filtered probability space $(\Omega,\mathcal{F},\{\mathcal{F}_t\}_{t\geq 0},\mathbb{P})$ and an asset price process $\mathbf{S}=\{S_t\}_{t\geq 0}$,
as a~non-negative diffusion process started at $S_0=S\in\R_+$ adapted to its natural filtration $\{\mathcal{F}_t\}_{t\geq 0}$ and
having the infinitesimal generator
\begin{equation}
   \label{Generator}
   (\mathcal{G}_{ _{\mathbf{S}}} f)(S) \deq \frac{1}{2}\,\sigma^2(S) f''(S) + (\mathrm{r}-\mathrm{q}) S f'(S)
\end{equation}
acting on a~bounded twice continuously differentiable function $f:\R_+\to\R$. Here, $\mathrm{r}\geq 0$ and $\mathrm{q}\geq 0$ are
respectively the constant risk-free interest rate and the constant dividend yield rate.
We will denote $\mathbb{P}$ as an equivalent martingale measure (i.e. risk-neutral probability measure) with discounted process
$\{\e^{-(\mathrm{r}-\mathrm{q})t} S_t\}_{t\geq 0}$ as a~$\mathbb{P}$-martingale. In all pricing models considered in this paper this property holds true.
The asset price obeys the stochastic differential equation
$\d S_t = (\mathrm{r}-\mathrm{q})\d t + \sigma(S_t)\d W_t$, where $\{W_t\}_{t\geq 0}$
is a~standard $\mathbb{P}$--Brownian motion. In the sequel, we shall simply set the dividend yield $\mathrm{q}=0$.

The boundary behavior of the $\mathbf{S}$-process at the origin and infinity depends on the growth behavior of $\sigma(S)$ as $S\to 0$ and $S\to\infty$,
respectively. For the models considered in this paper, $S=\infty$ is a~natural boundary. The point $S=0$ can be either natural, or a~regular boundary
(which we shall specify as killing) or an exit boundary. If the equity price process hits the zero boundary before the maturity date, the equity goes
to bankruptcy and a~derivative on the equity becomes worthless. Therefore, for a~model with default, the payoff in equation (\ref{StepPayoff}) takes the form
\[ f_\text{step}^\pm =\e^{-\rho A_{T,\mathbf{S}}^{L,\pm}} f(S_T) \indfunc_{T<\tau_0},\]
where $\tau_0$ is the first hitting time at zero. We note that this is equivalent to defining the payoff by equation (\ref{StepPayoff})
where $f(\partial)\equiv 0$ and the $\mathbf{S}$-process is given by the cemetery state $\partial$ upon hitting the origin, i.e.
$S_T \equiv \partial$ for $T \ge \tau_0$.

It is a~typical situation when a~solvable model can be obtained from another solvable underlying process, say a~diffusion $\mathbf{X}$,
by simply applying a~change of variables. Such a~transformation simplifies our formulation and it proves convenient to work directly with the
underlying diffusion $\mathbf{X}=\{X_t\}_{t\geq 0}$ instead of the asset price process $\mathbf{S}$.
We assume that the asset price process $\mathbf{S}$ is given by a~strictly increasing smooth
$\mathcal{C}^2(\I)$ map $\F:\I\to\R_+$ of the process $\mathbf{X}$,
i.e. $S_t=\F(X_t), t\ge 0$. The process $\mathbf{X}$ is taken to be a~one-dimensional time-homogeneous regular diffusion on
an interval $\I\equiv (l,r)$, $-\infty\leq l<r\leq\infty$. We denote the inverse map by $\X\equiv \F^{-1}$.
Since $\F$ is strictly increasing, the origin $S=0$ and the point at infinity $S=\infty$
have the same boundary classification as the respective left, $l$, and right, $r$, endpoints of $\mathbf{X}$.
Moreover, the occupation times for processes $\mathbf{S}$ and $\mathbf{X}$ are simply related:
\begin{equation*}
A_{t,\mathbf{S}}^{L,+} = A_{t,\mathbf{X}}^{\ell,+} \deq \int_0 ^t \indfunc_{X_u > \ell}\,\d u\,\,; \,\,
A_{t,\mathbf{S}}^{L,-} = A_{t,\mathbf{X}}^{\ell,-} \deq \int_0 ^t \indfunc_{X_u \le \ell}\,\d u,
\end{equation*}
where $\F(\ell) = L$, $\ell = \X(L)$ for any level $\ell\in\I$, i.e. $l < \ell < r$. The payoff in equation (\ref{StepPayoff}) is then equivalently
given in terms of the occupation times for $\mathbf{X}$ and its terminal value:
$f_\text{step}^\pm  = \exp ( -\orho A_{T,\mathbf{X}}^{\ell,\pm} )h(X_T)$ where $h(x)\deq f(\F(x))$. We note also that the event corresponding to the
$\mathbf{S}$-process hitting zero is equivalent to the $\mathbf{X}$-process hitting the left boundary $l$, at which time the process
$\mathbf{X}$ is sent to the cemetery state. Hence, for an exit or regular killing boundary $l$ we have $X_T \equiv \partial$ for $T \ge \tau_0$ and
we set $h\equiv 0$.

Throughout we respectively denote the (risk-neutral) expectation operator and probability measure for $\mathbf{X}$ started at $X_0 = x$
by $\E_x$ and $\P_{x}$. By the map we have the spot $S \equiv S_0 = \F(x)$ and $x = \X(S)$.
The no-arbitrage prices $V_\text{step}^\pm$ of proportional step options (with constant interest rate and zero dividend) take the form
\begin{align}
V_\text{step}^\pm (S,T)
   = \e^{-\mathrm{r}T}\E\left[f_\text{step}^\pm \mid S_0 = S\right] &= \e^{-\mathrm{r}T} \E_x \left[  \e^{-\orho A_{T,\mathbf{X}}^{\ell,\pm} } h(X_T) \right]
   \nonumber\\
   &= \e^{-\mathrm{r}T} \int_l^r h(y) \,\tilde{p}_\alpha^{\ell,\pm}(T;x,y) \,\d y.
   \label{step_price}
\end{align}
Hence, the problem of pricing proportional step options is reduced to evaluating an integral involving the transition PDF
$\tilde{p}_\alpha^{\ell,+}$ for the process $\mathbf{X}$ additionally killed according to its
($\alpha$-proportional) occupation time {\it above} level $\ell$ or
$\tilde{p}_\alpha^{\ell,-}$ for killing according to its occupation time {\it below} level $\ell$.

Let $\tau\sim \text{Exp}(\lambda)$ be
an exponentially distributed stopping time with parameter $\lambda>0$, independent of the process $\mathbf{X}$. For $\lambda > 0$,
the Green's functions give us expressions for the expectations of important related functionals involving the process
at the exponentially stopped time $\tau$:
\begin{align}
\E_x\left[ \e^{-\alpha A_{\tau,\mathbf{X}}^{\ell,+}}\,;\,X_\tau \in \d y\right]
\equiv \E_x\left[ \exp\left(-\alpha\int_0^\tau \indfunc_{X_u \ge \ell}\, \d u\right)\,;\,X_\tau \in \d y\right]
 = \lambda\widetilde{G}_\alpha^{\ell,+}(x,y,\lambda)\,\d y \,,\label{expectations-for-stopping1}\\
\E_x\left[ \e^{-\alpha A_{\tau,\mathbf{X}}^{\ell,-}}\,;\,X_\tau \in \d y\right]
\equiv \E_x\left[ \exp\left(-\alpha\int_0^\tau \indfunc_{X_u \le \ell}\, \d u\right)\,;\,X_\tau \in \d y\right]
 = \lambda\widetilde{G}_\alpha^{\ell,-}(x,y,\lambda)\,\d y\,.\label{expectations-for-stopping2}
\end{align}
The transition PDFs, $\tilde{p}_\alpha^{\ell,\pm}(t;x,y)$, $x,y\in\I, t\ge 0$, are given by the Laplace inverse (with respect to complex $\lambda$)
of the respective Green's functions:
\begin{align}\label{PDF-Green-Laplace}
\tilde{p}_\alpha^{\ell,\pm}(t;x,y)\,\d y  =  \E_x\left[ \e^{-\alpha A_{t,\mathbf{X}}^{\ell,\pm}}\,;\,X_t \in \d y\right]
&= {\mathcal L}_\lambda^{-1}\left(\lambda^{-1}\E_x\left[ \e^{-\alpha A_{\tau,\mathbf{X}}^{\ell,\pm}}\,;\,X_\tau \in \d y\right]\right)(t)
\nonumber \\
&= {\mathcal L}_\lambda^{-1}\big(\widetilde{G}_\alpha^{\ell,\pm}(x,y,\lambda)\big)(t) \,\d y\,.
\end{align}
The Green's functions $\widetilde{G}_\alpha^{\ell,\pm}(x,y,\lambda)
= \int_0^\infty \e^{-\lambda t} \tilde{p}_\alpha^{\ell,\pm}(t;x,y)\,\d t
 \deq {\mathcal L}_t\big(\tilde{p}_\alpha^{\ell,\pm}(t;x,y) \big)(\lambda)$ are given by the Laplace transform of the respective transition PDFs.
Note that, throughout our paper, the Green's functions and transition PDFs are defined {\it with respect to the Lebesgue measure}.
In what follows we firstly derive closed-form analytical expressions for $\widetilde{G}_\alpha^{\ell,\pm}(x,y,\lambda)$,
and hence for the expectations in (\ref{expectations-for-stopping1}) and (\ref{expectations-for-stopping2}).
This is the result in Lemma \ref{Lemma_Greenfunc}.
Secondly, we proceed to analytically invert the Laplace transform in (\ref{PDF-Green-Laplace}) and obtain closed-form spectral expansions
for $\tilde{p}_\alpha^{\ell,\pm}(t;x,y)$ in Proposition \ref{propn_spectral}.
We note that our results are valid for quite general diffusions.
The transition PDFs are later used to compute the option prices given by the integral in (\ref{step_price}).

At this point we remark on the important connection between the above transition PDF and the joint PDF of
the occupation time $A_{t,\mathbf{X}}^{\ell,\pm} \in [0,t]$ and the terminal value $X_t\in\I$ for the process started at $X_0=x$.
In particular, using the joint PDF (for either occupation time above or below level $\ell$)
\[\P_x\left(A_{t,\mathbf{X}}^{\ell,\pm}\in \d u \,,\, X_t\in\d y\right) := p_{A_t,X_t}^{\ell,\pm}(u,y|x) \d u\,\d y
\]
the expectation in (\ref{PDF-Green-Laplace}) takes the equivalent form
\[ \E_x\left[ \e^{-\alpha A_{t,\mathbf{X}}^{\ell,\pm}}\,;\,X_t \in \d y\right] =
\left( \int_0^t \e^{-\alpha u} p_{A_t,X_t}^{\ell,\pm}(u,y|x)\,\d u\right)\,\d y
\equiv {\mathcal L}_{u} \left(p_{A_t,X_t}^{\ell,\pm}(u,y|x) \right)\!(\alpha)\,\d y. \]
Hence, $p_{A_t,X_t}^{\ell,\pm}(u,y|x)$ and $\tilde{p}_\alpha^{\ell,\pm}(t;x,y)$ are Laplace transforms of one another. In particular,
\begin{equation} \label{SingleLaplacetransformPDF}
  p_{A_t,X_t}^{\ell,\pm}(u,y|x) = {\mathcal L}_\alpha^{-1} \left(  \tilde{p}_\alpha^{\ell,\pm}(t;x,y) \right)(u).
\end{equation}
Alternatively, this joint PDF can also be expressed as a double inverse Laplace transform of the Green's function:
\begin{equation} \label{DblLaplacetransformGreen}
 p_{A_t,X_t}^{\ell,\pm}(u,y|x)
  = {\mathcal L}_\alpha^{-1} \left( {\mathcal L}_\lambda^{-1}\big(\widetilde{G}_\alpha^{\ell,\pm}(x,y,\lambda)\big)(t)\right)(u).
\end{equation}
Since we are able to analytically invert the Laplace transform with respect to $\lambda$,
the joint density is given by a single Laplace inverse of the transition PDF $\tilde{p}_\alpha^{\ell,\pm}(t;x,y)$ as in (\ref{SingleLaplacetransformPDF}).
This single Laplace inversion operation can be readily performed numerically. Being given the joint PDF $p_{A_t,X_t}$, the no-arbitrage price of a general derivative contract whose payoff is a function of the occupation time and the terminal asset price, $f(A_{T,\mathbf{S}},S_T)$, can be computed as follows:
\begin{align}
V(S,T)
   &= \e^{-\mathrm{r}T}\E\left[f(A_{T,\mathbf{S}},S_T) \mid S_0 = S\right] = \e^{-\mathrm{r}T} \E_x \left[ f(A_{T,\mathbf{X}},\F(X_T)) \right]
   \nonumber\\
   &= \e^{-\mathrm{r}T} \int_0^T\int_l^r f(u,y) \,p_{A_T,X_T}(u,y|x) \,\d u \,\d y
   \label{general_price}\\
   &= \e^{-\mathrm{r}T} \int_0^T\int_l^r f(u,y) \,{\mathcal L}_\alpha^{-1} \left(  \tilde{p}_\alpha^{\ell,\pm}(T;x,y) \right)\!(u) \,\d u \,\d y. \nonumber
\end{align}
Here we also note that this double integral reduces to the single integral in (\ref{step_price}) in the particular case of pricing step options.
Clearly, equation (\ref{step_price}) gives the more direct method for computing step options, both analytically and numerically.
If we only had the double Laplace transform of the PDF $p_{A_t,X_t}$, the calculation of the derivative price would involve possibly a two-dimensional integral  in the double Laplace inverse to obtain $p_{A_t,X_t}$ and then the two-dimensional integral over time and space.


We now proceed with the details of the diffusion processes that we consider in this paper. We begin by assuming that the underlying diffusion $\mathbf{X}$ has the infinitesimal generator
\begin{equation}
   \label{Generator_X}
   (\G f)(x) \deq \frac{1}{2} b^2(x) f''(x)+a(x)f'(x)
\end{equation}
acting on a~bounded twice continuously differentiable function $f:\I\to\R$.
The (infinitesimal) drift and diffusion coefficient functions are assumed smooth with continuous
$a(x), a'(x), b(x) > 0, b''(x)$ on the open interval $\I$. The diffusion $\mathbf{X}$ hence has smooth
scale and speed density functions respectively defined by
\begin{equation}
\s(x)\deq \exp\left(-\int^x\frac{2a(z)}{b^2(z)}dz\right)
   \mbox{ \ and \ } \m(x)\deq\frac{2}{b^2(x)\s(x)}. \label{smx}
   \end{equation}
We recall that the Green's function $G(x,y,\lambda)$ for process $\mathbf{X}$ (with generator in (\ref{Generator_X})) has the standard form
(e.g. see \cite{BS02}):
\begin{equation}
G(x,y,\lambda) = ({\mathcal W}_\lambda)^{-1}\m(y)\psi_\lambda(x \wedge y)\phi_\lambda(x\vee y),
\label{greenfunc}
\end{equation}
$x \wedge y\deq\min\{x,y\}$ and $x\vee y \deq\max\{x,y\}$. The pair of functions $\{\psi_\lambda,\phi_\lambda\}$
solve $\G \varphi(x) = \lambda \varphi(x)$. Their Wronskian is given by
$W[\phi_\lambda,\psi_\lambda](x) = {\mathcal W}_\lambda \,\s(x)$ with ${\mathcal W}_\lambda$ depending only on $\lambda$. [Throughout, the
Wronskian of two functions is denoted by $W[f,g](x)\deq f(x)g'(x) - g(x)f'(x)$.] These functions are uniquely characterized
(within a~multiplicative constant) by requiring that, for real values of $\lambda>0$,
$\psi_\lambda$ and $\phi_\lambda$ are respectively increasing and decreasing functions on $\I$
and by additionally posing boundary conditions at regular (non-singular) boundaries of $X$ (see~\cite{BS02}).
For a~regular left boundary $l$, $\psi_\alpha(l+) = 0$ if $l\notin \I$
is specified as killing or $\frac{1}{\s(l+)}\frac{d\psi_\alpha(l+)}{dx}=0$
if $l$ is specified as reflecting and included in the state space.
If $l$ is a~singular boundary, the functions have the following boundary properties:
if $l$ is entrance($\equiv$entrance-not-exit), then
$\psi_\alpha(l+) > 0, \frac{1}{\s(l+)}\frac{d\psi_\alpha(l+)}{dx}=0$;
if $l$ is exit($\equiv$exit-not-entrance), then $\psi_\alpha(l+) = 0,
\frac{1}{\s(l+)}\frac{d\psi_\alpha(l+)}{dx}>0$;
if $l$ is a~natural boundary, then $\psi_\alpha(l+) =0,
\frac{1}{\s(l+)}\frac{d\psi_\alpha(l+)}{dx}=0$.
Analogous conditions hold for the right boundary:
$\psi_\alpha(r-) > 0, \frac{1}{\s(r-)}\frac{d\psi_\alpha(r-)}{dx}=0$ if $r$ is entrance;
$\psi_\alpha(r-) = 0$, $\frac{1}{\s(r-)}\frac{d\phi_\alpha(r-)}{dx} < 0$ if $r$ is exit;
$\psi_\alpha(r-) = 0$, $\frac{1}{\s(r-)}\frac{d\phi_\alpha(r-)}{dx} = 0$ if $r$ is a~natural boundary.

In what follows we shall deal with Wronskians of fundamental solutions with differing values of the eigenvalue parameter. Hence,
it is convenient to adopt a~slightly more compact notation as follows. Let $\lambda, \gamma \in\C$, then we define:
\begin{equation}\label{Wronskians}
\wron{\phi}{\phi}{\lambda}{\gamma}(x) \deq W[\phi_\lambda,\phi_\gamma](x),\quad \wron{\phi}{\psi}{\lambda}{\gamma}(x) \deq W[\phi_\lambda,\psi_\gamma](x),\quad \wron{\psi}{\psi}{\lambda}{\gamma}(x) \deq W[\psi_\lambda,\psi_\gamma](x),
\end{equation}
where $\wron{\psi}{\phi}{\gamma}{\lambda}(x) = - \wron{\phi}{\psi}{\lambda}{\gamma}(x)$ and for $\gamma=\lambda$
we recover the above Wronskian $\wron{\phi}{\psi}{\lambda}{\lambda}(x) = {\mathcal W}_\lambda \,\s(x)$.

\begin{lemma} \label{Lemma_Greenfunc}
Let $\mathbf{X}$ have the Green's function in (\ref{greenfunc}) with
$\psi_\lambda$ and $\phi_\lambda$ specified as above. Then, the Green's function
$\widetilde{G} = \widetilde{G}_\alpha^{\ell,+}(x,y,\lambda)$, solving
$(\tilde{\G}^+_x -\lambda)\tilde{G} = -\delta (x - y)$ with infinitesimal generator $\tilde{\G}^+_x \deq
(\G - \alpha\indfunc_{x\ge\ell})$, $\alpha\in\R$, and having boundary conditions as in $G$, is given by:
\begin{equation}\label{Greenfunc_plus}
\widetilde{G}_\alpha^{\ell,+}(x,y,\lambda) = \left\{\begin{array}{ll}
 G(x,y,\lambda) + \displaystyle\frac{\wron{\phi}{\phi}{\lambda}{\lambda+\alpha}(\ell)}{\wron{\phi}{\psi}{\lambda+\alpha}{\lambda}(\ell)}
  \m(y){\psi_\lambda(x)\psi_\lambda(y) \over {\mathcal W}_\lambda}, & x \leq \ell, \,y \leq \ell,\\\\
    \displaystyle{\s(\ell)\over \wron{\phi}{\psi}{\lambda+\alpha}{\lambda}(\ell)}
   \m(y)\psi_\lambda(x)\phi_{\lambda+\alpha}(y), & x \leq \ell, \,y \geq \ell,\\\\
    \displaystyle{\s(\ell)\over \wron{\phi}{\psi}{\lambda+\alpha}{\lambda}(\ell)}
   \m(y)\phi_{\lambda+\alpha}(x)\psi_\lambda(y), & x \geq \ell, \,y \leq \ell,\\\\
   G(x,y,\lambda+\alpha) + \displaystyle\frac{\wron{\psi}{\psi}{\lambda}{\lambda+\alpha}(\ell)}{\wron{\phi}{\psi}{\lambda+\alpha}{\lambda}(\ell)}
  \m(y){\phi_{\lambda+\alpha}(x)\phi_{\lambda+\alpha}(y) \over {\mathcal W}_{\lambda+\alpha}}, & x \geq \ell, \,y \geq \ell.
\end{array}\right.
\end{equation}
Silimarly, $\widetilde{G} = \widetilde{G}_\alpha^{\ell,-}(x,y,\lambda)$ solving
$(\tilde{\G}^-_x -\lambda)\tilde{G} = -\delta (x - y)$ with infinitesimal generator $\tilde{\G}^-_x \deq
(\G - \alpha\indfunc_{x\le \ell})$, $\alpha\in\R$, and having boundary conditions as in $G$, is given by:
\begin{equation}\label{Greenfunc_minus}
\widetilde{G}_\alpha^{\ell,-}(x,y,\lambda) = \left\{\begin{array}{ll}
   G(x,y,\lambda+\alpha) + \displaystyle\frac{\wron{\phi}{\phi}{\lambda+\alpha}{\lambda}(\ell)}{\wron{\phi}{\psi}{\lambda}{\lambda+\alpha}(\ell)}
  \m(y){\psi_{\lambda+\alpha}(x)\psi_{\lambda+\alpha}(y) \over {\mathcal W}_{\lambda+\alpha}}, & x \leq \ell, \,y \leq \ell, \\\\
    \displaystyle{\s(\ell)\over \wron{\phi}{\psi}{\lambda}{\lambda+\alpha}(\ell)}
   \m(y)\psi_{\lambda+\alpha}(x)\phi_\lambda(y), & x \leq \ell, \,y \geq \ell,\\\\
    \displaystyle{\s(\ell)\over \wron{\phi}{\psi}{\lambda}{\lambda+\alpha}(\ell)}
   \m(y)\phi_{\lambda}(x)\psi_{\lambda+\alpha}(y), & x \geq \ell, \,y \leq \ell,\\\\
    G(x,y,\lambda) + \displaystyle\frac{\wron{\psi}{\psi}{\lambda+\alpha}{\lambda}(\ell)}{\wron{\phi}{\psi}{\lambda}{\lambda+\alpha}(\ell)}
  \m(y){\phi_\lambda(x)\phi_\lambda(y) \over {\mathcal W}_\lambda}, &  x \geq \ell, \,y \geq \ell.
\end{array}\right.
\end{equation}
\end{lemma}
\begin{proof} See Appendix \ref{subsect_A11}.
\end{proof}
\noindent Note that in the trivial case where $\alpha = 0$ we recover the Green's function for process $\mathbf{X}$ on $\I$, i.e.
$\widetilde{G}_0^{\ell,\pm}(x,y,\lambda) = G(x,y,\lambda)$. In the limit $\alpha\to\infty$, it can be shown that the above Green's
functions respectively recover those for the process killed at a~lower or upper level $\ell$
(see the remark just after the proof of Lemma \ref{Lemma_Greenfunc}).

An analytical inversion of the respective Laplace transforms
leads to closed-form spectral expansions for the transition PDFs. The Green's functions in equations (\ref{Greenfunc_plus}) and (\ref{Greenfunc_minus}) are
functions of complex $\lambda\in\C$ and the respective Laplace inverses are given by the Bromwich contour integral:
$\tilde{p}_\alpha^{\ell,\pm}(t;x,y) = {1\over 2\pi i}\int_{c - i\infty}^{c + i\infty} \e^{\lambda t} \widetilde{G}_\alpha^{\ell,\pm}(x,y,\lambda) \,\d \lambda$,
where all singularities of the respective Green's functions lie to the left of the Bromwich line $\text{Re}\,\lambda = c$.
The spectral expansion for $\tilde{p}_\alpha^{\ell,\pm}(t;x,y)$ can be obtained by appropriately closing the
Bromwich contour and applying Cauchy's Residue Theorem. The residue contributions from simple poles give rise to the discrete part of the
spectral expansion, while continuous parts of the spectral expansion can arise as integrals over branch cut discontinuities of $\widetilde{G}$
in the complex $\lambda$ plane to the left of the Bromwich line. The analytical form of the spectral expansion
of $\tilde{p}_\alpha^{\ell,\pm}(t;x,y)$ clearly depends upon the singularity structure of
$\widetilde{G}_\alpha^{\ell,\pm}(x,y,\lambda)$.

If the Green's function is a~meromorphic function that is analytic in $\lambda$ with the
exception of a~countable number of isolated simple poles then the corresponding transition PDF has a~discrete spectral expansion.
In this case we are in so-called {\it Spectral Category I}. We refer to \cite{Linet04}
for a~general discussion and summary of the possible spectral categories and their relation to the boundary classification of the endpoints
for a~one-dimensional time-homogeneous diffusion. The spectral category and the analytic properties of the
fundamental pair $\{\psi_\lambda,\phi_\lambda\}$ are intimately related
to the boundary classification. We remark that the boundary classification (i.e. natural, exit, entrance or regular) of the endpoints $l,r$ is the
same for process $\mathbf{X}$, with generator $\G$, as for the processes $\tilde{\mathbf{X}}^{\ell,\pm}$ defined by the respective
generators $\tilde{\G}^\pm$ in Lemma \ref{Lemma_Greenfunc}.
This fact is easily shown since all processes have the same above scale and speed measure and the additional killing rate,
$\alpha \indfunc_{x\ge \ell}$ or $\alpha \indfunc_{x\le \ell}$,
is a~piecewise constant function of $x$, and hence does not affect the Feller conditions. Moreover, the two Sturm-Liouville (SL)
operators $-\tilde{\G}^\pm$ and $-\G$ differ only by a~piecewise constant function $\alpha \indfunc_{x\ge \ell}$ or $\alpha \indfunc_{x\le \ell}$.
The so-called potential function in the Liouville normal form of the SL equation associated to the respective sets of operators
$-\tilde{\G}^\pm$ and $-\G$ also differs only by a~piecewise constant function.
It follows that an endpoint $e\in\{l,r\}$ is non-oscillatory (NONOSC) for the process $\mathbf{X}$ if and only if it is NONOSC
for processes $\tilde{\mathbf{X}}^{\ell,\pm}$. Moreover, an endpoint $e\in\{l,r\}$ that is O-NO (oscillatory/non-oscillatory)
for $\mathbf{X}$ will also be O-NO for processes $\tilde{\mathbf{X}}^{\ell,\pm}$ with spectral cutoff that may be shifted by the amount $\alpha$.

If both endpoints $\{l,r\}$ are NONOSC, the eigenspectrum of the SL operator is simple (non-negative for $\alpha \ge 0$) and purely discrete,
i.e. we are in Spectral Category I. Non-natural (regular, exit or entrance) boundaries are always NONOSC, while natural boundaries may
be NONOSC or O-NO with some spectral cutoff value. In what follows we shall assume that Spectral Category I holds where the endpoints $\{l,r\}$ of the
diffusion $\mathbf{X}$ are NONOSC. The Green's functions $\widetilde{G}_\alpha^{\ell,\pm}(x,y,\lambda)$ in Lemma \ref{Lemma_Greenfunc}
are meromorphic functions of $\lambda$ with simple poles at $\lambda = -\tilde{\lambda}_n, n =1,2,\ldots$, where the set of eigenvalues $\{\tilde{\lambda}_n\}_{n\ge 1}$ corresponding
to the SL operator $-\tilde{\G}^+$ (or $-\tilde{\G}^-$) form a~monotonically increasing sequence of real values (non-negative for $\alpha \ge 0$):
$\tilde{\lambda}_n \nearrow \infty$ as $n \nearrow \infty$.
The following result gives us the spectral expansions for $\tilde{p}_\alpha^{\ell,\pm}(t;x,y)$ in the general case of Spectral Category I.
\begin{prop} \label{propn_spectral}
Assume Spectral Category I with endpoints of the diffusion $\mathbf{X}$ as NONOSC. Then, the
Green functions $\widetilde{G}_\alpha^{\ell,\pm}(x,y,\lambda)$ in (\ref{Greenfunc_plus}) and (\ref{Greenfunc_minus})
are meromorphic functions of $\lambda$ and the corresponding transition PDFs (i.e. Laplace transform inverses) have the discrete spectral expansions
\begin{equation}\label{spec_exp}
\tilde{p}_\alpha^{\ell,\pm}(t;x,y) = \m(y)\sum_{n=1}^\infty \e^{-\tilde\lambda_n t}\tilde\phi_{n,\alpha}^{\ell,\pm}(x)\tilde\phi_{n,\alpha}^{\ell,\pm}(y).
\end{equation}
For $\tilde{p}_\alpha^{\ell,+}$:
\begin{equation}\label{spec_exp_a}
\tilde\phi_{n,\alpha}^{\ell,+}(x)\tilde\phi_{n,\alpha}^{\ell,+}(y) = \left\{\begin{array}{ll}
\left( \displaystyle\frac{\wron{\phi}{\,\phi}{-\tilde\lambda_n}{-\tilde\lambda_n +\alpha}(\ell)}
{C^{\ell,+}_{n,\alpha}\,{\mathcal W}_{\!-\tilde\lambda_n}}\right)
 \psi_{\!-\tilde\lambda_n}(x) \psi_{\!-\tilde\lambda_n}(y), & x \leq \ell, \,y \leq \ell,\\\\
 \displaystyle{ \frac{\s(\ell)}{C^{\ell,+}_{n,\alpha}}}
 \psi_{\!-\tilde\lambda_n}(x) \phi_{\!-\tilde\lambda_n + \alpha}(y), & x \leq \ell, \,y \geq \ell,\\\\
 \displaystyle{ \frac{\s(\ell)}{C^{\ell,+}_{n,\alpha}}}
 \phi_{\!-\tilde\lambda_n + \alpha}(x)\psi_{\!-\tilde\lambda_n}(y),  & x \geq \ell, \,y \leq \ell,\\\\
\left(\displaystyle\frac{\wron{\psi}{\,\psi}{-\tilde\lambda_n}{-\tilde\lambda_n +\alpha}(\ell)}
 {C^{\ell,+}_{n,\alpha}\,{\mathcal W}_{\!-\tilde\lambda_n + \alpha}}\right)
 \phi_{\!-\tilde\lambda_n + \alpha}(x) \phi_{\!-\tilde\lambda_n + \alpha}(y), & x \geq \ell, \,y \geq \ell.
\end{array}\right.
\end{equation}
where $C^{\ell,+}_{n,\alpha} \deq {d\over d\lambda}\wron{\phi}{\,\psi}{\lambda +\alpha}{\lambda}(\ell)\vert_{\lambda=-\tilde\lambda_n}$
with eigenvalues $\{\tilde\lambda_n\}_{n\ge 1}$ as the set of increasing simple zeros
solving $\wron{\phi}{\,\psi}{-\tilde\lambda_n +\alpha}{-\tilde\lambda_n}(\ell) = 0$.

\noindent For $\tilde{p}_\alpha^{\ell,-}$:
\begin{equation}\label{spec_exp_b}
\tilde\phi_{n,\alpha}^{\ell,-}(x)\tilde\phi_{n,\alpha}^{\ell,-}(y) = \left\{\begin{array}{ll}
\left( \displaystyle\frac{\wron{\phi}{\,\phi}{-\tilde\lambda_n +\alpha}{-\tilde\lambda_n}(\ell)}
{C^{\ell,-}_{n,\alpha}\,{\mathcal W}_{\!-\tilde\lambda_n + \alpha}}\right)
 \psi_{\!-\tilde\lambda_n + \alpha}(x) \psi_{\!-\tilde\lambda_n + \alpha}(y), & x \leq \ell, \,y \leq \ell,\\\\
 \displaystyle{ \frac{\s(\ell)}{C^{\ell,-}_{n,\alpha}}}
 \psi_{\!-\tilde\lambda_n +\alpha}(x) \phi_{\!-\tilde\lambda_n}(y), & x \leq \ell, \,y \geq \ell,\\\\
 \displaystyle{ \frac{\s(\ell)}{C^{\ell,-}_{n,\alpha}}}
 \phi_{\!-\tilde\lambda_n}(x)\psi_{\!-\tilde\lambda_n +\alpha}(y),  & x \geq \ell, \,y \leq \ell,\\\\
\left(\displaystyle\frac{\wron{\psi}{\,\psi}{-\tilde\lambda_n +\alpha}{-\tilde\lambda_n}(\ell)}
 {C^{\ell,-}_{n,\alpha}\,{\mathcal W}_{\!-\tilde\lambda_n}}\right)
 \phi_{\!-\tilde\lambda_n}(x) \phi_{\!-\tilde\lambda_n}(y), & x \geq \ell, \,y \geq \ell.
\end{array}\right.
\end{equation}
where $C^{\ell,-}_{n,\alpha} \deq {d\over d\lambda}\wron{\phi}{\,\psi}{\lambda}{\lambda +\alpha}(\ell)\vert_{\lambda=-\tilde\lambda_n}$
with eigenvalues $\{\tilde\lambda_n\}_{n\ge 1}$ as the set of increasing simple zeros
solving $\wron{\phi}{\,\psi}{-\tilde\lambda_n}{-\tilde\lambda_n +\alpha}(\ell) = 0$.
\end{prop}
\begin{proof} See Appendix \ref{subsect_A12}.
\end{proof}

By using the spectral expansions in Proposition \ref{propn_spectral} within the integral in (\ref{step_price}),
we are now able to compute the no-arbitrage prices of various proportional step options.
In this paper, we are interested in pricing proportional step options with call and put payoffs, where
$f(S_T)$ is either the vanilla call payoff $(S_T - K)^+$ or the vanilla put payoff $(K - S_T)^+$.
However, our method can be used for any well-behaved payoff function $f$.
From equation (\ref{step_price}), the respective pricing of the vanilla step (up/down) call and put options with level $L$ are given by the integrals
\begin{align}
C^\pm_{\mathrm{step}}(S,T) & = \e^{-\mathrm{r}T} \int_{x_{_K}}^r (\F(y) - K) \,\tilde{p}_\alpha^{\ell,\pm}(T;x,y) \,\d y, \label{CallValue} \\
P^\pm_{\mathrm{step}}(S,T) & = \e^{-\mathrm{r}T} \int_l^{x_{_K}} (K - \F(y)) \,\tilde{p}_\alpha^{\ell,\pm}(T;x,y) \,\d y, \label{PutValue}
\end{align}
where $x_{_K} = \X(K), x = \X(S), \ell = \X(L), l = \X(0), r = \X(\infty)$, $\X \equiv \F^{-1}$.
Note that the deltas (and gammas) of step options are also readily computed. For instance,
by differentiating equation (\ref{step_price}) with respect to $S$ we obtain the delta:
\begin{equation} \label{Delta}
\Delta_{\mathrm{step}}^{\pm} \equiv \frac{\partial}{\partial S} V_\text{step}^\pm (S,T)
                    = \e^{-\mathrm{r}T} \X'(S) \int_l^r f(\F(y))\frac{\partial}{\partial x} \tilde{p}_\alpha^{\ell,\pm}(T;x,y)  \,\d y .
\end{equation}

As an alternative to spectral expansions, one can make use of the numerical Laplace inversion when computing the no-arbitrage prices of step options. Such an approach leads to computing the following two-dimensional integral:
\begin{align*}
   V_\text{step}^\pm (S,T)
   &= \e^{-\mathrm{r}T} \int_l^r h(y) \,{\mathcal L}_\lambda^{-1}\big(\widetilde{G}_\alpha^{\ell,\pm}(x,y,\lambda)\big)(T) \,\d y\\
   &= \e^{-\mathrm{r}T} \int_l^r h(y) \,\left[ \frac{1}{2\pi i} \int_{c-i\infty}^{c+i\infty} \e^{\lambda t} \,\widetilde{G}_\alpha^{\ell,\pm}(x,y,\lambda)\,\d\lambda  \right] \d y \text{ \ \ (or)}\\
   &= \e^{-\mathrm{r}T} \frac{2}{\pi}\e^{cT}\int_l^r h(y) \,\left[ \int_0^\infty \cos(ut) \,\Re\left\{ \widetilde{G}_\alpha^{\ell,\pm}(x,y,c+iu)\right\}\,\d u  \right] \d y,
\end{align*}
where all singularities are to left of the Bromwich line $\Re \lambda =c$. The numerical inversion of the Laplace transform can be done by employing the Euler numerical algorithm. However, the Green's functions of models considered in this paper are formed of hypergeometric functions which, for complex arguments, are quite complicated to compute. This circumstance makes the numerical Laplace inversion possible, yet not as simple a computational procedure. The use of spectral expansions, which are uniform rapidly convergent series, gives us the added advantage of computing $\tilde{p}_\alpha^{\ell,\pm}$ to very high accuracy and efficiency, especially for larger values of $T$.


\section{Pricing under Solvable Models} \label{sec:4}

According to Proposition \ref{propn_spectral}, to apply the spectral expansion method and find no-arbitrage prices of step options under a~diffusion asset price model (specified by the infinitesimal generator $\G$ and defined on state space $\I$), we only need to know two
linearly independent fundamental solutions, $\phi_\lambda(x)$ and $\psi_\lambda(x)$, that solve the equation $(\G\, f)(x) = \lambda f(x)$,
$\lambda\in \C$, $x\in\I$,  subject to respective boundary conditions at the endpoints of $\I$. Other ingredients such as derivatives and
Wronskians of the fundamental solutions can be computed analytically or numerically via finite differences. Since the zeros
$\{\tilde{\lambda}_n\}_{n\geq 1}$ w.r.t. $\lambda$ of either Wronskian $\wron{\phi}{\psi}{-\lambda + \alpha}{-\lambda}(\ell)$ or
$\wron{\phi}{\psi}{-\lambda}{-\lambda + \alpha}(\ell)$, which can be computed numerically, converge to the zeros
$\{\lambda_n\}_{n\geq 1}$ of $\mathcal{W}_{-\lambda} = W[\phi_{-\lambda}, \psi_{-\lambda}](x) / \s(x)$ in the
limit $\alpha\to 0+$, we are also interested in the distribution of $\{\lambda_n\}_{n\geq 1}$. First, these zeros
can be used as initial guesses for computing $\{\tilde{\lambda}_n\}_{n\geq 1}$. Second, we analyze whether the zeros
grow linearly or quadratically. In the latter case, the spectral expansion series converges more rapidly than in the
former case and fewer terms are required to achieve a~high accuracy of computations.

In the sequel, we present four analytically solvable asset price models with state-dependent volatility functions.
The first model to be considered is the well-known constant elasticity of variance (CEV) diffusion model. Three other
alternative models are constructed by using the \textit{``diffusion canonical transformation''} (see \cite{CM08,CM10,CM11} for details).
In particular, we consider three examples of hypergeometric diffusion models namely the Bessel-K, confluent-U, and UOU models respectively
constructed from the squared Bessel, CIR, and Ornstein-Uhlenbeck processes by using the aforementioned method. For all four models considered
in the sequel, the discounted asset price process $\{\e^{-\mathrm{r}t} S_t\}_{t\geq 0}$ is a~martingale under the risk-neutral probability
measure $\mathbb{P}$ where we assume zero dividend on the stock.

\subsection{The CEV Diffusion Model}  \label{subsec:4.1}
The constant elasticity of variance (CEV) diffusion, $\mathbf{S}=\left\{S_t\right\}_{t\geq 0}\in\R_+$ is defined by the infinitesimal generator
$(\mathcal{G}_{_{\mathbf{S}}}\, f)(S) \equiv \frac{1}{2}\,\delta^2 S^{2\beta + 2} f''(S) + \mathrm{r} Sf'(S)$
with $\delta>0$ and $\mathrm{r}\in\R$. Here we take $\mathrm{r} > 0$ and assume $\beta < 0$. Hence, the point $S=\infty$ is a~ natural boundary.
For $\beta < -1/2$, the point $S=0$ is a~regular boundary, which we specify as killing, and for $-1/2 \le \beta < 0$ it is an exit boundary.

Recall that the Cox-Ingersoll-Ross (CIR) model (known also as the squared radial Ornstein-Uhlenbeck process, see \cite{BS02})
has the infinitesimal generator
\begin{equation} \label{GenCIR}
(\G\, f)(x) = \frac{1}{2}\nu^2 x f''(x) + (\gamma_0-\gamma_1 x)f'(x),
\end{equation}
with constant parameters $\gamma_0$, $\gamma_1$, and $\nu>0$. The strictly increasing mapping
$\textsf{X}(S) \equiv \delta^{-2} \beta^{-2} S^{-2\beta}$ reduces the~CEV process to the CIR model with the parameters
$\nu=2$, $\gamma_0=2+\frac{1}{\beta}$, and $\gamma_1=2\mathrm{r}\beta.$ It is convenient to define the parameter
$\mu\equiv \frac{2\gamma_0}{\nu^2}-1=\frac{\gamma_0}{2}-1 = {1\over 2\beta}$.
The respective speed and scale densities are $\m(x) = {1\over 2}x^\mu \e^{-\gamma_1 x/2}$ and $\s(x) = x^{-\mu-1} \e^{\gamma_1 x/2}$.
The left endpoint $l=0$ is regular killing for $\mu \in (-1,0)$ and is exit for $\mu \le -1$;
the right endpoint $r=\infty$ is NONOSC natural. Hence, both endpoints are NONOSC and Proposition~\ref{propn_spectral} is applicable.

The respective fundamental solutions for the CIR process are
\begin{equation} \label{CIRfs}
  \begin{split}
  \phi_\lambda(x)& =x^{|\mu|}\e^{\gamma_1 x/2}\,\U\left(1+\frac{\lambda}{|\gamma_1|},1+|\mu|,{|\gamma_1|\over 2} x\right),\\
  \psi_\lambda(x)&=x^{|\mu|}\e^{\gamma_1 x/2}\,\M\left(1+\frac{\lambda}{|\gamma_1|},1+|\mu|,{|\gamma_1|\over 2} x\right),
  \end{split}
\end{equation}
where $\M(a,b,z)$ and $\U(a,b,z)$ are the confluent hypergeometric functions (see \cite{AbrSteg}).
From known analytic properties of $M$ and $U$, we note that $\psi_\lambda(x)$ and $\phi_\lambda(x)$ are entire in $\lambda$.
These fundamental solutions satisfy the Wronskian relation $\wron{\phi}{\psi}{\lambda}{\lambda}(x) = \mathcal{W}_\lambda \s(x)$ with
\begin{equation} \label{Wron_CEV}
\mathcal{W}_\lambda = \frac{\Gamma(1+|\mu|)}{\Gamma\left(1+\frac{\lambda}{|\gamma_1|}\right)}.
\end{equation}
The corresponding Green's function in equation (\ref{greenfunc}) is meromorphic with simple poles given by the
simple zeros of $\mathcal{W}_\lambda$: $\lambda = -\lambda_n= -|\gamma_1|n$, $n=1,2,3,\ldots$. The zeros $\{\lambda_n\}_{n\ge 1}$ can be used as initial
guesses for finding the eigenvalues $\{\tilde\lambda_n\}_{n\ge 1}$. In particular, for $\tilde{p}_\alpha^{\ell,+}$
(or $\tilde{p}_\alpha^{\ell,-}$) these zeros are initial guesses for the zeros $\{\tilde\lambda_n\}_{n\ge 1}$ (w.r.t. $\lambda$) of
$\wron{\phi}{\psi}{-\lambda + \alpha}{-\lambda}(\ell)$ (or $\wron{\phi}{\psi}{-\lambda}{-\lambda + \alpha}(\ell)$). To compute the
Wronskians within the spectral expansion formula (\ref{spec_exp_a}) we require formulae for the derivatives of the fundamental functions in (\ref{CIRfs})
with respect to $x$. Such derivatives are obtained by using the differential recurrences: ${d\over dz}\M(a,b,z) = (a/b)\M(a+1,b+1,z)$ and
${d\over dz}\U(a,b,z) = -a \U(a+1,b+1,z)$.

\subsection{Alternative Diffusion Models}  \label{subsec:4.2}

The diffusion canonical transformation is defined as a~combination of a~change of measure and a~nonlinear mapping. Consider a~{\it solvable underlying time-homogeneous diffusion}, say $\mathbf{X}^{(0)}\equiv \{X^{(0)}_t\}_{t\ge 0}$,
defined on the state space $\I=(l,r)$, $-\infty\leq l<r\leq \infty$, and specified
by smooth drift, $a_0(x)$, and diffusion, $b(x)$, coefficients and two fundamental solutions
$\psi^{(0)}_\lambda\equiv \Psi_\lambda$ and $\phi^{(0)}_\lambda\equiv \Phi_\lambda$ of
$(\G^{(0)} f)(x) \deq \frac{b^2(x)}{2} f''(x) + a_0(x) f'(x) = \lambda f(x)$, $\lambda\in \C$, $x\in\I$, subject to appropriate boundary conditions.
$\Psi_\rho$ and $\Phi_\rho$ and are respectively increasing and decreasing positive functions of $x\in\I$ for real values of the parameter
$\rho>0$. By applying the change of measure, the solvable underlying diffusion is transformed into another diffusion
process\footnote{Such a~process $\mathbf{X}$ and its generator were denoted more explicitly by $\{X^{(\rho)}_t\}_{t\ge 0}$ and $\G^{(\rho)}$
in our previous papers \cite{CM08,CM10,CM11}. To avoid excessive use of notation, and to maintain consistency in this paper, we drop the superscript $\rho$
and write $X^{(\rho)}_t\equiv X_t$ and $\G^{(\rho)} \equiv \G$. Likewise we write $\m_\rho \equiv \m$ and $\s_\rho \equiv \s$.
The underlying process is denoted here by $\mathbf{X}^{(0)}$ and has respective speed and scale functions $\m_0$ and $\s_0$.}
$\mathbf{X}=\{X_t\in \I\}_{t\ge 0}$  with generator as in (\ref{Generator_X}):
\begin{equation} \label{Generator_rho}
(\G\,f)(x)\equiv\frac{1}{2}b^2(x)f^{\prime\prime}(x)
 +\underbrace{\left(a_0(x)+b^2(x)\frac{u'_\rho(x)}{u_\rho(x)}\right)}_{=a(x)}f^{\prime}(x)\,,
\end{equation}
where $u_\rho(x)\equiv q_1 \Psi_\rho(x) + q_2 \Phi_\rho(x)$ is a~strictly positive function with constants $q_{1,2}\geq 0,$ $q_1+q_2>0$.
For a~more detailed general discussion of this transformation to various solvable diffusions we refer to \cite{CM08,CM10}.
We simply note here that this corresponds to a~(time-homogeneous) Doob-$h$ transform which is generated by the so-called ($\rho$-excessive)
generating function which is here given by $h(x)\equiv u_\rho(x)$.

A transition density $p$ of the diffusion $\mathbf{X}$ relates to a~transition density
$p_0$ of an underlying diffusion $\mathbf{X}^{(0)}$ as follows:
\begin{equation}
p(t;x,y)=\e^{-\rho t}\frac{u_\rho(y)}{u_\rho(x)}
p_0(t;x,y),\;x,y\in\I\,,\;t>0\,. \label{XrhoPDF}
\end{equation}
The speed, $\m(x)$, and scale, $\s(x)$, densities for $\mathbf{X}$ are obtained from the
the speed, $\m_0(x)$, and scale, $\s_0(x)$, densities of the underlying diffusion via $\m(x) = \m_0(x)u^2_\rho(x)$ and $\s(x) = \s_0(x)/u^2_\rho(x)$.
The above relationship follows similarly when additional killing is introduced.
In particular, the transition PDFs $\tilde{p}^\pm\equiv \tilde{p}_\alpha^{\ell,\pm}$, for the diffusions $\tilde{\mathbf{X}}^{\ell,\pm}$,
are related to transition PDFs $\tilde{p}_0^\pm \equiv \tilde{p}_{0,\alpha}^{\ell,\pm}$ for underlying diffusions $\tilde{\mathbf{X}}^{(0)\ell,\pm}$,
defined by the generators $\tilde{\G}^{(0)+}_x \deq (\G^{(0)} - \alpha\indfunc_{x\ge\ell})$ and
$\tilde{\G}^{(0)-}_x \deq (\G^{(0)} - \alpha\indfunc_{x\le\ell})$, as
\begin{equation}
\tilde{p}^\pm(t;x,y)=\e^{-\rho
t}\frac{u_\rho(y)}{u_\rho(x)}
\tilde{p}_0^\pm(t;x,y). \label{XrhoPDF_2}
\end{equation}
The main important point of this formula is that closed-form transition PDFs, $\tilde{p}^\pm$,
are obtained directly from the closed-form transition PDFs, $\tilde{p}_0^\pm$, for the solvable underlying process.
The latter are given by the spectral expansion formulae in Proposition~\ref{propn_spectral} {\it applied to the simpler
analytically solvable underlying process} $\mathbf{X}^{(0)}$,
i.e. in all terms in equations (\ref{spec_exp}), (\ref{spec_exp_a}) and (\ref{spec_exp_b}), we make the obvious replacements:
$\psi\to \Psi, \phi\to\Phi, \m\to\m_0, \s\to\s_0, \wron{\phi}{\phi}{\lambda}{\lambda +\alpha}\to \wron{\Phi}{\Phi}{\lambda}{\lambda +\alpha},
\wron{\phi}{\psi}{\lambda}{\lambda +\alpha}\to \wron{\Phi}{\Psi}{\lambda}{\lambda +\alpha},
\wron{\phi}{\psi}{\lambda +\alpha}{\lambda}\to \wron{\Phi}{\Psi}{\lambda +\alpha}{\lambda},
\wron{\psi}{\psi}{\lambda}{\lambda +\alpha}\to \wron{\Psi}{\Psi}{\lambda}{\lambda +\alpha},
{\mathcal W}_\lambda \equiv \wron{\phi}{\psi}{\lambda}{\lambda}(x)/\s(x) \to
{\mathcal W}^{(0)}_\lambda \equiv \wron{\Phi}{\Psi}{\lambda}{\lambda}(x)/\s_0(x)$. Correspondingly,
the eigenvalues $\{\tilde\Lambda_n\}_{n\ge 1}$ are defined as the set of increasing simple zeros
(w.r.t. $\lambda$) of $\wron{\Phi}{\Psi}{-\lambda +\alpha}{-\lambda}(\ell)$ for
$\tilde{p}_0^+$ and of $\wron{\Phi}{\Psi}{-\lambda}{-\lambda +\alpha}(\ell)$ for $\tilde{p}_0^-$.
For instance, in case $x\le \ell,y \le \ell$ we have:
\begin{equation}
\tilde{p}_{0,\alpha}^{\ell,+}(t;x,y) = \m_0(y)\sum_{n=1}^\infty \e^{-\tilde\Lambda_n t}
\left( \displaystyle\frac{\wron{\Phi}{\,\Phi}{-\tilde\Lambda_n}{-\tilde\Lambda_n +\alpha}(\ell)}
{C^{\ell,+}_{n,\alpha}{\mathcal W}^{(0)}_{\!-\tilde\Lambda_n}}\right)
\Psi_{\!-\tilde\Lambda_n}(x) \Psi_{\!-\tilde\Lambda_n}(y)
\end{equation}
where $C^{\ell,+}_{n,\alpha} \deq {d\over d\lambda}\wron{\Phi}{\,\Psi}{\lambda +\alpha}{\lambda}(\ell)\vert_{\lambda=-\tilde\Lambda_n}$
with eigenvalues $\{\tilde\Lambda_n\}_{n\ge 1}$ as the set of increasing simple zeros (w.r.t. $\lambda$) of
$\wron{\Phi}{\,\Psi}{-\lambda +\alpha}{-\lambda}(\ell)$.

We remark that the formula in (\ref{XrhoPDF_2}) is also readily derived by making use of the fact that a~pair of fundamental solutions, $\psi_\lambda^{(\rho)}\equiv \psi_\lambda$
and $\phi_\lambda^{(\rho)}\equiv\phi_\lambda$ for process $\mathbf{X}$ with generator in (\ref{Generator_rho}), for any given $\rho > 0$,
is given by ratios of the known fundamental
solutions for the underlying diffusion:
\begin{equation}
  \label{FundSolXrho}
  \phi_\lambda(x) = \frac{\Phi_{\lambda+\rho}(x)}{u_\rho(x)},\quad
  \psi_\lambda(x) = \frac{\Psi_{\lambda+\rho}(x)}{u_\rho(x)}.
\end{equation}
Equation (\ref{XrhoPDF_2}) then follows from (\ref{spec_exp_a}) upon applying the basic Wronskian property:
$W\left[\frac{f}{u},\frac{g}{u}\right](x) = \frac{W[f,g](x)}{u^2(x)}.$
The eigenvalues are found by obtaining the zeros $\{\tilde\Lambda_n\}_{n\ge 1}$ and then adding
$\rho$ to them: $\tilde{\lambda}_n = \tilde\Lambda_n + \rho$, $n\geq 1$,
i.e. the eigenspectrum is simply shifted by the positive amount $\rho$ via the Dooh-$h$ transform as seen in (\ref{XrhoPDF_2}).

Finally, the second main step is to obtain a~solvable diffusion $\mathbf{S}=\{S_t \equiv \F
(X_t), t\ge 0\}$, which is used here as an asset price model. This process is defined by a~strictly monotonic real-valued mapping
$\F$ with $\F ', \F ''$ continuous on $\I$. The mapping $\F$ admits the general quotient form:
 \begin{equation} \label{FMAP}
 \F (x) = \frac{c_1 \Psi_{\rho+\mathrm{r}}(x)+c_2
 \Phi_{\rho+\mathrm{r}}(x)}{u_\rho(x)},\quad c_1,c_2\in\R.
 \end{equation}
The infinitesimal generator of the process $\mathbf{S}$ is given by
\begin{equation} \label{GenF}
(\mathcal{G}_{ _{\mathbf{S}}}\, f)(S) \equiv \frac{1}{2}\,\sigma^2(S)
f''(S) + \mathrm{r}Sf'(S),
\end{equation}
where $S\in\I_S=\left(\min\{\mathsf{F}(l+),\mathsf{F}(r-)\},\max\{\mathsf{F}(l+),\mathsf{F}(r-)\}\right)$. The diffusion coefficient (volatility) function is
\begin{equation} \label{sigmaF}
 \sigma(S)=b(x)|\F'(x)|=\frac{b(x)|W[u_\rho, c_1 \Psi_{\rho+\mathrm{r}} + c_2 \Phi_{\rho+\mathrm{r}}](x)|}{u^2_\rho(x)}\,,\quad x=\X (S)\,,
\end{equation}
and $\mathrm{r}$ is a~real constant such that $\rho+\mathrm{r}>0$. The parameter $\mathrm{r}$ is equal to the risk-free positive interest rate.
As above, $\X \equiv\F^{-1}$ denotes the inverse map.

\subsubsection{The Bessel-K Model with Killing at an Upper Boundary}  \label{subsubsec:4.2.1}
Here we specifically consider a~$4$-parameter Bessel $\mathsf{K}$-family that arises from an underlying
($\gamma_0$-dimen\-sional) squared Bessel process (SQB), where we shall assume positive parameters $\mu\equiv\frac{2\gamma_0}{\nu^2}-1>0$ and $\nu>0$.
By applying the Doob transform with $u_\rho(x)\equiv x^{-\mu/2}K_{\mu}\left(2\sqrt{2\rho x}/\nu\right)$
to the SQB process $\mathbf{X}^{(0)}\in \R_+$ we obtain a~diffusion $\mathbf{X}\in \R_+$ with generator
$(\G\,f)(x)\equiv\frac{1}{2}\nu^2 x f^{\prime\prime}(x) + \left( \gamma_0+\nu^2 x \frac{u'_\rho(x)}{u_\rho(x)} \right) f^{\prime}(x)$.
The SQB process has speed and scale densities $\m_0(x) = {2\over \nu^2}x^\mu$ and $\s_0(x) = x^{-\mu-1}$.
For the Bessel-K model the mapping in (\ref{FMAP}) is the strictly increasing function
\begin{equation} \label{FmapBesK}
  \F(x)= c\displaystyle\frac{I_{\mu} \left( 2\sqrt{2(\rho+\mathrm{r}) x}/\nu\right)}{K_{\mu} \left(2\sqrt{2\rho x}/\nu\right)},
\end{equation}
where $c$ and $\rho$ are independently adjustable positive parameters.
$\F$ (and its inverse $\X$) maps $x\in(0,\infty)$ and $S\in(0,\infty)$ into one another.
The functions $I_{\mu}$ and $K_{\mu}$ denote the modified Bessel functions (of order $\mu$) of the first
and second kind, respectively (see \cite{AbrSteg} for definitions and properties).

For such a~Bessel $\mathsf{K}$-family of processes on $\R_+$, the origin is regular for $\mu\in (0,1)$ and exit for $\mu \ge 1$ (i.e. NONOSC) and
the point at infinity is O-NO natural. Hence, to guarantee that we have a~process with NONOSC endpoints (i.e. so that we are in
Spectral Category I and Proposition~\ref{propn_spectral} applies)
we introduce additional killing at some upper level $\h  >0$ and consider $\mathbf{X}=\{X_t\}_{t\ge 0}\in (0,\h )$.
Hence, the transformation (\ref{FmapBesK}) leads to a~family of processes
$\{S_t=\F(X_t)\}_{t\geq 0}\in (0,\H )$, $\H =\F(\h )$.
The boundary $S=0$ is exit if $\mu\geq 1$ or is a~regular (specified as killing) boundary if $0<\mu<1$;
the boundary $S=\H =\F(\h )$ is a~killing boundary.
We note that one way to deal with the process on $S_t\in\R_+$ is to consider the limiting case
where the upper boundary is progressively increased to infinity. As $\H \to\infty$ (i.e. $\h \to\infty$), the spectral expansions of the
transition PDFs, and hence the prices of proportional step options, converge to those for the Bessel $\mathsf{K}$-family on $\R_+$.

Differentiating (\ref{FmapBesK}), and applying differential relations $zI^{\prime}_\mu(z) = \mu I_\mu(z) + zI_{\mu + 1}(z)$
and $zK^{\prime}_\mu(z) = \mu K_\mu(z) - zK_{\mu + 1}(z)$, gives the volatility function for the Bessel $\mathsf{K}$-family via equation (\ref{sigmaF}):
\begin{equation}\label{SigmaBesK}
  \sigma(S)= c\sqrt{2}\left(
     \textstyle\frac{\sqrt{\rho}\,I_\mu\left(2\sqrt{2(\rho + \mathrm{r}) x}/\nu\right)K_{\mu+1}\left(2\sqrt{2\rho x}/\nu\right)}
{K_{\mu}^2\left(2\sqrt{2\rho x}/\nu\right)}
+ \textstyle\frac{\sqrt{\rho + \mathrm{r}}\,I_{\mu+1}\left(2\sqrt{2(\rho + \mathrm{r})x}/\nu\right)}{K_{\mu}\left(2\sqrt{2\rho x}/\nu\right)} \right),
\end{equation}
where $x=\X(S) = \F^{-1}(S)$. The fundamental solutions in (\ref{FundSolXrho}) follow from those for the SQB process
with the above mentioned boundary conditions at the left and right endpoints $l=0$ and $r=\h $:
\begin{equation} \label{BKfs}
\begin{split}
  \Psi_\lambda(x) &= x^{-\mu/2}I_{\mu}\big(2\sqrt{2\lambda x}/\nu\big),\\
  \Phi_\lambda(x) &= x^{-\mu/2}[I_{\mu}\big(2\sqrt{2\lambda \h }/\nu\big)K_{\mu}\big(2\sqrt{2\lambda x}/\nu\big)
   - K_{\mu}\big(2\sqrt{2\lambda \h }/\nu\big)I_{\mu}\big(2\sqrt{2\lambda x}/\nu\big)].
\end{split}
\end{equation}
These solutions satisfy the Wronskian relation $\wron{\Phi}{\Psi}{\lambda}{\lambda}(x) =
\frac{1}{2}I_{\mu}\big(2\sqrt{2\lambda \h }/\nu\big)\s_0(x)$.
Other Wronskians $\wron{\Phi}{\Psi}{\lambda}{\lambda + \alpha}, \wron{\Phi}{\Phi}{\lambda}{\lambda + \alpha},
\wron{\Psi}{\Psi}{\lambda}{\lambda + \alpha}$ required in the spectral expansion formulas are also easily computed by applying the above differential relations.
To avoid computations of Bessel functions of complex argument, we use the following well-known identities:
\begin{equation*}
  I_\mu(ix) = i^\mu J_\mu(x),\quad
  K_\mu(ix)I_\mu(iy) - I_\mu(ix)K_\mu(iy) = \frac{\pi}{2}\left(J_\mu(x)Y_\mu(y) - Y_\mu(x)J_\mu(y)\right),\\
\end{equation*}
where $J_\mu$ and $Y_\mu$ are the (ordinary) Bessel functions (of order $\mu$) of the first and second kind, respectively.
Firstly, the zeros $\{\Lambda_n\}_{n\ge 1}$ of $\wron{\Phi}{\Psi}{-\lambda}{-\lambda}$ (w.r.t. $\lambda$) can be computed numerically.
These zeros are all positive and grow quadratically. In particular, $\Lambda_n = (\nu^2/8\h )j^2_{n,\mu}$ where
$j_{n,\mu}$ are the positive simple zeros of the ordinary Bessel function, i.e. $J_\mu(j_{n,\mu}) = 0$.
The set $\{\Lambda_n\}_{n\ge 1}$ is then used as initial guess for computing the eigenvalues $\{\tilde{\Lambda}_n\}_{n\ge 1}$ which are the
simple zeros (w.r.t. $\lambda$) for either Wronskian $\wron{\Phi}{\Psi}{-\lambda + \alpha}{-\lambda}(\ell)$ or
$\wron{\Phi}{\Psi}{-\lambda}{-\lambda + \alpha}(\ell)$. Combining all quantities gives us $\tilde{p}_{0,\alpha}^{\ell,\pm}$ and hence
$\tilde{p}_{\alpha}^{\ell,\pm}$ by (\ref{XrhoPDF_2}).

\subsubsection{The Confluent-U Model}  \label{subsubsec:4.2.2}

The confluent-U family of diffusions arises by considering the CIR diffusion as the underlying process $\mathbf{X}^{(0)}\in \R_+$
with generator $\G^{(0)}$ given by (\ref{GenCIR}). Although this family can be defined for a~larger set of parameters,
here we shall assume positive parameters $\gamma_0,\gamma_1$ and define
$\upsilon \equiv \frac{\rho}{\gamma_1}$, $\mu\equiv\frac{2\gamma_0}{\nu^2}-1>0$, $\kappa \equiv \frac{2\gamma_1}{\nu^2}$.
The speed and scale densities for the CIR process are $\m_0(x) = (\kappa/\gamma_1)x^\mu \e^{-\kappa x}$ and $\s_0(x) = x^{-\mu - 1}\e^{\kappa x}$.
Applying the Doob transform with generating function $u_\rho(x)\equiv\U(\upsilon,\mu+1,\kappa x)$ to the CIR process gives
us a~diffusion process $\mathbf{X}=\{X_t\}_{t\ge 0}\in \R_+$ with generator
$(\G\,f)(x)\equiv\frac{1}{2}\nu^2 x f^{\prime\prime}(x) + \left( \gamma_0-\gamma_1 x+\nu^2 x \frac{u'_\rho(x)}{u_\rho(x)} \right) f^{\prime}(x)\,,$
where $\rho > 0$. For the confluent-U family of models, the map $\F$ is given
by the strictly increasing map
\begin{equation} \label{FmapCIRU}
\F(x)=c\frac{\M(\upsilon+\frac{\mathrm{r}}{\gamma_1},\mu+1,\kappa x)}{\U(\upsilon,\mu+1,\kappa x)}
\end{equation}
where $c>0$. Differentiating (\ref{FmapCIRU}), while using ${d\over dz}\M(a,b,z) = (a/b)\M(a+1,b+1,z)$ and
${d\over dz}\U(a,b,z) = -a \U(a+1,b+1,z)$, the volatility function in (\ref{sigmaF})
for the confluent-U family of processes $\{S_t=\F(X_t)\}\in \R_+$ takes the form
\begin{equation} \label{SigmaCIRU}
  \begin{split}
    \sigma(S) &= c\kappa\nu\sqrt{x}\bigg(
    \displaystyle\frac{\upsilon \,\M(\frac{\rho+\mathrm{r}}{\gamma_1},\mu+1,\kappa x)\,\U(\upsilon + 1 ,\mu+2,\kappa x)} {\U^2(\upsilon,\mu+1,\kappa x)}\\ & +
    \displaystyle\frac{(\frac{\rho+\mathrm{r}}{\gamma_1})\,\M(\frac{\rho+\mathrm{r}}{\gamma_1}+1,\mu+2,\kappa x)}{(\mu + 1)\,\U(\upsilon,\mu+1,\kappa x)}\bigg),
  \end{split}
\end{equation}
where $x=\X(S) = \F^{-1}(S)$. For the confluent-U model the origin $S=0$ is exit if $\mu\geq 1$ and regular killing if $0<\mu<1$;
the point at infinity is natural. The fundamental solutions in (\ref{FundSolXrho}) follow from those for the CIR process:
\begin{equation} \label{CUfs}
\Psi_\lambda(x) = \M(\lambda/\gamma_1,\mu+1,\kappa x), \,\,\,\,
\Phi_\lambda(x) = \U(\lambda/\gamma_1,\mu+1,\kappa x).
\end{equation}

For the confluent-U model, both endpoints of the state space $(0,\infty)$ are NONOSC
so that Spectral Category I holds and Proposition~\ref{propn_spectral} applies.
The functions in (\ref{CUfs}) satisfy the Wronskian relation
$\wron{\Phi}{\Psi}{-\lambda}{-\lambda}(x) / \s_0(x) = \mathcal{W}^{(0)}_{-\lambda}$ with
\begin{equation} \label{Wron_CU}
\mathcal{W}^{(0)}_{-\lambda} = \frac{\kappa^{-\mu}\Gamma(1 + \mu)}{\Gamma\left(-\lambda/\gamma_1\right)}.
\end{equation}
The zeros $\{\Lambda_n\}_{n\ge 1}$ (w.r.t. $\lambda$) are simply given by the poles of the Gamma function in the denominator. Hence, the eigenvalues
grow linearly and are given by $\Lambda_n = \gamma_1(n-1),\, n=1,2,\ldots$.
The set $\{\Lambda_n\}_{n\ge 1}$ is then used as initial guess for computing the eigenvalues $\{\tilde{\Lambda}_n\}_{n\ge 1}$ which are the
simple zeros (w.r.t. $\lambda$) for either Wronskian $\wron{\Phi}{\Psi}{-\lambda + \alpha}{-\lambda}(\ell)$ or
$\wron{\Phi}{\Psi}{-\lambda}{-\lambda + \alpha}(\ell)$. The Wronskians $\wron{\Phi}{\Psi}{\lambda}{\lambda + \alpha}, \wron{\Phi}{\Phi}{\lambda}{\lambda + \alpha},
\wron{\Psi}{\Psi}{\lambda}{\lambda + \alpha}$ required in the spectral expansion formulas are computed by applying the
above differential relations for the confluent $M$ and $U$ functions.
Combining all quantities gives us $\tilde{p}_{0,\alpha}^{\ell,\pm}$ and finally $\tilde{p}_{\alpha}^{\ell,\pm}$ by (\ref{XrhoPDF_2}).

\subsubsection{The UOU Model} \label{subsubsec:4.2.3}

We now consider the regular Ornstein-Uhlenbeck (OU) process $\mathrm{X}^{(0)} \in (-\infty,\infty)$ with constant diffusion coefficient
$b(x)=\nu>0$ and linear drift coefficient $a_0(x) = -\gamma_1 x$, $\gamma_1>0$. The fundamental solutions for this OU process are
$$\Phi_\lambda(x)=\e^{\kappa x^2/4}D_{-\lambda/\gamma_1}(x\sqrt{\kappa})
\text{ and } \Psi_\lambda(x)=\Phi_\lambda(-x) = \e^{\kappa x^2/4}D_{-\lambda/\gamma_1}(-x\sqrt{\kappa})$$
where $D_{\nu}(\cdot)$ is Whittaker's parabolic cylinder function (see \cite{AbrSteg}).

We now apply the diffusion canonical transformation to the OU process with choice of parameters $q_1 = 1,q_2=0$, i.e. with generating function
$u_\rho(x) = \Psi_\rho(x) = \e^{\kappa x^2/4}D_{-\rho/\gamma_1}(x\sqrt{\kappa})$ to
obtain the process $\mathbf{X}=\{X_t\in\R\}_{t\geq 0}$ having the generator $(\G\,f)(x)\equiv\frac{1}{2}\nu^2 f^{\prime\prime}(x)
+ \left( -\gamma_1 x+\nu^2  \frac{u'_\rho(x)}{u_\rho(x)} \right) f^{\prime}(x)$. The choice of function in (\ref{FMAP}) given by
\begin{equation} \label{FmapUOU}
\mathsf{F}(x)=c\frac{\Psi_{\mathrm{r}+\rho}(x)}{\Phi_\rho(x)}=
c\frac{D_{- \frac{\rho + \mathrm{r}}{\gamma_1}}(-x\sqrt{\kappa})}{D_{-\frac{\rho}{\gamma_1}}(x\sqrt{\kappa})}\,, \,\,\,c>0\,,
\end{equation}
maps $x\in\R$ onto $S\in(0,\infty)$ and is monotonically increasing. This transformation gives us a~family of asset price processes with generator in
(\ref{GenF}) that is referred to as the unbounded Ornstein-Uhlenbeck (UOU) model with the diffusion coefficient function given by
\begin{equation} \label{SigmaUOU}
 \sigma(S) = c\nu\sqrt{\kappa} \left\{\frac{\rho + \mathrm{r}}{\gamma_1} \frac{D_{-(1 + \frac{\rho + \mathrm{r}}{\gamma_1})}(-x\sqrt{\kappa})}
{D_{-\frac{\rho}{\gamma_1}}(x\sqrt{\kappa})}
+ \frac{\rho}{\gamma_1} \frac{D_{-\frac{\rho + \mathrm{r}}{\gamma_1}}(-x\sqrt{\kappa})}{D_{-\frac{\rho}{\gamma_1}}(x\sqrt{\kappa})}
\frac{D_{-(1 + \frac{\rho}{\gamma_1})}(x\sqrt{\kappa})}
{D_{-\frac{\rho}{\gamma_1}}(x\sqrt{\kappa})}\right\},
\end{equation}
where $x=\mathsf{X}(S)\equiv \F^{-1}(S)$. Both endpoints, $S=0$ and $S=\infty$, of the UOU process are NONOSC natural boundaries.
Hence, Spectral Category I holds and Proposition~\ref{propn_spectral} applies.

The fundamental functions above satisfy the Wronskian relation
$\wron{\Phi}{\Psi}{-\lambda}{-\lambda}(x) / \s_0(x) = \mathcal{W}^{(0)}_{-\lambda}$ where
\begin{equation} \label{Wron_OU}
\mathcal{W}^{(0)}_{-\lambda} = \frac{\sqrt{2\kappa\pi}}{\Gamma(-\lambda/\gamma_1)}
\end{equation}
has simple zeros $\{\Lambda_n\}_{n\ge 1}$ given by $\Lambda_n = \gamma_1(n-1),\, n=1,2,\ldots$.
The set $\{\Lambda_n\}_{n\ge 1}$ is then used as initial guess for computing the eigenvalues $\{\tilde{\Lambda}_n\}_{n\ge 1}$ which are the
simple zeros (w.r.t. $\lambda$) for either Wronskian $\wron{\Phi}{\Psi}{-\lambda + \alpha}{-\lambda}(\ell)$ or
$\wron{\Phi}{\Psi}{-\lambda}{-\lambda + \alpha}(\ell)$. The other Wronskians
$\wron{\Phi}{\Psi}{\lambda}{\lambda + \alpha}, \wron{\Phi}{\Phi}{\lambda}{\lambda + \alpha},
\wron{\Psi}{\Psi}{\lambda}{\lambda + \alpha}$ are computed by applying the
differential relation ${d\over dz}D_{-\upsilon}(z) = -(z/2)D_{-\upsilon}(z) - \upsilon D_{-\upsilon - 1}(z)$.
Combining all quantities gives us $\tilde{p}_{0,\alpha}^{\ell,\pm}$ and hence $\tilde{p}_{\alpha}^{\ell,\pm}$ via (\ref{XrhoPDF_2}).

\section{Numerical Evaluation of Step Option Prices} \label{sec:5}

\subsection{Spectral Series Expansions} \label{subsec:5.1}
In this section, we discuss the computational details for computing step option prices using (\ref{step_price}) and the
spectral expansion formula for the transition PDFs $\tilde{p}\equiv \tilde{p}_\alpha^{\ell,\pm}(T;x,y)$.
Given a~discount factor $\alpha$, a~level $L$ and spot $S$ (hence $\ell = \X(L)$ and $x = \X(S)$), the PDF in~(\ref{step_price}) is to be computed for varying values of $y\in [y_{\min},y_{\max}]$; $\tilde{p}$ is computed on such values by truncating the spectral expansion in (\ref{spec_exp}) to the first $N$ terms, where $N$ is sufficiently large:
\begin{equation}\label{spec_expN}
\tilde{p}_\alpha^{\ell,\pm}(T;x,y) \approx \m(y)\sum_{n=1}^N \e^{-\tilde\lambda_n T}\tilde\phi_{n,\alpha}^{\ell,\pm}(x)\tilde\phi_{n,\alpha}^{\ell,\pm}(y).
\end{equation}

The numerical procedure consists of the following basic steps.
First, compute numerically the first $N$ eigenvalues, $\{\tilde{\lambda}_n\}_{1\leq n\leq  N}$ associated to the process $\mathrm{X}$. For the solvable asset price processes arising from the Doob transform we simply compute the first $N$ eigenvalues $\{\tilde{\Lambda}_n\}_{1\leq n\leq  N}$ associated to the underlying process $\mathrm{X}^{(0)}$   and thereby obtain $\tilde{\lambda}_n = \tilde{\Lambda}_n + \rho$.  Note that the eigenvalues grow linearly (for the CEV, confluent-U, and UOU models) or quadratically (for the Bessel-K model with killing) as $n$ increases.
The computations of the terms in the spectral series expansions can be split in three parts. That is, we can individually compute the parts that only depend on $y$, $x$, and $\ell$, respectively. Partial derivatives of Wronskians with respect to $\lambda$ can be calculated numerically by using a~central finite difference approximation. Upon completing the evaluation of the spectral series, a~quadrature rule (e.g. the adaptive Simpson rule) is applied to compute the integral in~(\ref{step_price}). Clearly, this computational scheme can be easily parallelized at different stages.

For most of the above models, including the CEV, the computation of the spectral expansion for the transition PDF $\tilde{p}$
requires many evaluations of the confluent hypergeometric function $\U(a,b,z)$
for negative values of $a$. To avoid introducing numerical errors while computing $U$, the rescaled version of the confluent hypergeometric function is used:
\begin{align*}
   &\frac{\U(a+n,b,z)}{\Gamma(-a)} = \frac{\sin(\pi(b-a))}{\sin(\pi b)} \frac{(-1)^n}{\B(-a,b)}\frac{\Gamma(b-a-n)}{\Gamma(b-a)} \M(a+n,b,z) \\
   & + (-1)^n\frac{\sin(\pi a)}{\pi}\Gamma(b-1)\frac{\Gamma(1-n-a)}{\Gamma(-a)}z^{1-b} \M(a-b+1+n,2-b,z),
\end{align*}
where $a<0$, $a\not\in\mathds{Z}$, and $n=0,1,2,\ldots$. Here $\B$ denotes the Beta function.
The parabolic cylinder function $D_\nu$ can be expressed in terms of the Kummer confluent hypergeometric
function $\M$ provided that $\nu\not\in\mathds{Z}$. The rescaled version of $D$ is as follows:
\[ \frac{D_{\nu+n}(z)}{\Gamma(\nu/2)\cdot 2^{\nu/2}} = 2^{n/2} \e^{-z^2/4} \sqrt{\pi} \cdot\left( \frac{\M(-(\nu+n)/2,1/2,z^2/2)}{\Gamma(\nu/2)\Gamma((1-\nu-n)/2)} - \frac{\sqrt{2} z \M((1-\nu-n)/2,3/2,z^2/2)}{\Gamma(\nu/2)\Gamma(-(\nu+n)/2)} \right),\]
where $\nu>0$, $\nu\not\in\mathds{Z}$, and $n\in\mathds{Z}$. To compute the products of gamma functions (for large values of $\nu$) in the aforementioned formula, we use Euler's reflection formula $\Gamma(1-z)\Gamma(z)=\frac{\pi}{\sin(\pi z)}$ and the following asymptotic series:
\[ \frac{\Gamma(x+\frac{1}{2})}{\Gamma(x)} = \sqrt{x} \left( 1-\frac{1}{8x}+\frac{1}{128x^2}+\frac{5}{1024 x^3} - \frac{21}{32768 x^4} + \cdots\right).\]

Since numerical errors can be introduced when directly computing $\phi_\lambda(x)$ and respective Wronskians for large negative values of $\lambda$,
we can define a~rescaled version of $\phi_\lambda$ denoted by $\hat{\phi}_\lambda$ that has a~better asymptotic behaviour as $\lambda\to-\infty$.
For the CEV model we define
\[ \hat{\phi}_\lambda(x) = \frac{\phi_\lambda(x)}{\Gamma\left(-\frac{\lambda}{|\gamma_1|}\right)} = x^{|\mu|}\e^{\gamma_1 x/2}\frac{\U\left(1+\frac{\lambda}
{|\gamma_1|},1+|\mu|,|\gamma_1| x/2\right)}{\Gamma\left(-\frac{\lambda}{|\gamma_1|}\right)},\quad \lambda<0. \]
For the Confluent-$\U$ model and UOU model, we respectively set
\[ \hat{\phi}_\lambda(x) = \frac{\phi_\lambda(x)}{\Gamma(-\frac{\rho+\lambda}{\gamma_1})} \text{ and }
\hat{\phi}_\lambda(x) = \frac{\phi_\lambda(x)}{\Gamma(-\frac{\rho+\lambda}{2\gamma_1})\cdot 2^{(\rho+\lambda)/(2\gamma_1)}},\quad \lambda<-\rho. \]

\subsection{Monte Carlo Bridge Approximation} \label{subsec:5.2}

We can compare numerical values obtained by using the analytical
spectral expansions of the previous section with Monte Carlo approximation values. In~\cite{MW11}, a~novel algorithm for
the exact simulation of occupation times for a~Brownian bridge is constructed. The method is used to
approximately sample occupation times for a~nonlinear solvable diffusion that admits an exact path simulation.
Such an approximation sampling algorithm can be applied to the CEV model and other solvable hypergeometric diffusions (i.e. the Bessel-K, confluent-U, and UOU models) considered in this paper for which an exact path simulation algorithm is available (see~\cite{MG10}).
For example, consider the CEV asset price process $\mathbf{S}$. There exists a~strictly increasing mapping $\X$ that maps
$\mathbf{S}$ into the CIR diffusion $\mathbf{X}$ whose volatility is a~square-root function, $\nu\sqrt{x}$. The
increasing mapping $\mathsf{Y}(x)=\frac{2}{\nu}\sqrt{x}$ reduces the process $\mathbf{X}$ to a~diffusion $Y_t=\mathsf{Y}(X_t)$
whose diffusion coefficient is equal to one. Thus, we have $Y_t=\mathsf{Y}(\X(S_t))$ or $S_t=\F(\mathsf{Y}^{-1}(Y_t)) =
\F({\nu^2\over 4}Y^2_t)$ for all $t\geq 0$, where $\F=\X^{-1}$.

On short time intervals $[t_1,t_2]$ such a~diffusion pinned at points $y_1$ and $y_2$ at respective times
$t_1$ and $t_2$ can be approximated by a~Brownian bridge from $y_1$ to $y_2$ over $[t_1,t_2]$. Therefore,
occupation times of the $\mathbf{S}$ bridge process on short time intervals can be well approximated by Brownian
bridge occupation times. Again, we use the fact that a~monotone transformation of a~random process does not change
the occupations times: $A^{L,\pm}_{t,\mathbf{S}}=A^{\X(L),\pm}_{t,\X(\mathbf{S})}$.

Our approach for the approximate sampling of occupation times $A^{L,\pm}_{T,\mathbf{S}}$ works as follows.
\begin{enumerate}
\item By using an algorithm from \cite{MG10}, draw a~sample path $S_{t_1},\ldots,S_{t_M}$ for a~given time partition $\{ 0=t_1<t_2<\ldots<t_M=T \}$.
\item Obtain the respective sample path of the underlying process with unit diffusion coefficient by using the transformation
$Y_{t_i} = \textsf{Y}(\textsf{X}(S_{t_i}))$ for each $i=0,1,\ldots,M$.
\item Sample the occupation times of $A_{[t_{i-1},t_i]}^{\ell,\pm}$ for the Brownian bridge from $Y_{t_{i-1}}$ to $Y_{t_i}$ over $[t_{i-1},t_i]$ for each $i=1,\ldots,M$. Here, $\ell = \textsf{Y}(\textsf{X}(L))$.
\item Obtain the approximations $A^{L,\pm}_{T,\mathbf{S}} \approx \sum_{i=1}^M A_{[t_{i-1},t_i]}^{\ell,\pm}$.
\end{enumerate}
Note that the algorithms developed in \cite{MG10} allow us to simultaneously sample the first hitting time at zero, $\tau_0$, and a~sample path. If $\tau_0<T$, then the option is worthless. So the simulation of the occupation time can be skipped whenever $\tau_0<T$.

\subsection{Numerical Results} \label{subsec:5.3}

Let us compute prices and deltas of proportional step-down call and step-down put options under four different asset price models. The call and put payoff functions are respectively $\e^{-\alpha A^{L,-}_T} (S_T-K)_+\indfunc_{\tau_0>T}$ and
$\e^{-\alpha A^{L,-}_T}(K-S_T)_+\indfunc_{\tau_0>T}$  with proportionality factor $\alpha=5$, level $L=90$, and maturity time $T=\frac{1}{2}$.

Let us begin with the CEV model. It is known from \cite{JR98} that the negative elasticity values $\beta$ are typical for stock index options such as
S\&P~500. We use $\beta=-2$. The value of $\delta$ is selected so that the local (instantaneous) volatility $\sigma_0\equiv\sigma(S_0)/S_0$ at
the spot value $S_0=100$ equals 25\%: $\delta=\sigma_0 S_0^{-\beta}=0.25\cdot 100^{-\beta}$. For $\beta=-2$, we have $\delta=2500$. The CEV model is compared with the other hypergeometric
diffusion models with state-dependent volatility function $\sigma(S)$. The parameters of the models are adjusted so that the local volatility $\sigma(S_0)$ at spot $S_0=100$
is fixed at 25\%. The parameters of all four asset price models are summarized in Table~\ref{tb:Models}. It should be clear that the set used in Table~\ref{tb:Models}
is {\it not} the only choice giving $\sigma(S_0) = 25\%$. In fact, there is a~continuum of parameter sets
for which we can have a~fixed value for the local volatility. The different parameter values allow us to adjust the steepness,
skewness or smile features afforded by the various models. This is one important attractive feature of these models, particularly for calibration purposes.
Figure~\ref{fig:LocVol} illustrates the variety of typical
shapes of the local volatility functions $\sigma(S)/S$ when choosing one model over another for given choice of parameters in Table~\ref{tb:Models}.
Within a~given model, we can also further adjust the shapes by varying the model parameters.

\noindent\textbf{Test 1.}\\
First, we calculated the values of step-down call and put options with fixed spot $S_0=100$ and varied strike $K\in\{80,90,100,110,120\}$ under the four models (as given in Table~\ref{tb:Models}). The computations were done in Matlab with the use the QUADV routine, which numerically evaluates integrals using the recursive adaptive Simpson quadrature rule. The absolute error tolerance was set at $10^{-8}$. Computation of the integrals in (\ref{CallValue}) and (\ref{PutValue}) with several strikes $K$ reduces to the numerical evaluation of integrals of the following two forms:
\[ \int_{x_i}^{x_{i+1}} \F(y) \,\tilde{p}_\alpha^{\ell,\pm}(T;x,y) \,\d y \text{ \ and \ }\int_{x_i}^{x_{i+1}} \tilde{p}_\alpha^{\ell,\pm}(T;x,y) \,\d y \]
on several intervals $(x_i,x_{i+1})$. The results of the numerical tests are presented in Table~\ref{table1A}. Figure~\ref{fig:PDFs} provides typical shapes of the stock price PDF $p(S)\deq\tilde{p}_\alpha^{\ell,-}(T;x,\X(S))\X'(S)$ (where $T=\frac{1}{2}$, $x=\X(S_0=100)$, and $\ell=\X(L=90)$) computed for the four asset price models.

All computations were done on a Hewlett-Packard(R) Notebook PC with a four-core Intel(R) Core(TM) i7 CPU Q720 @ 1.6GHz and 4 GB of memory. Some details regarding the computational time are provided in Table~\ref{table1B}. Here we use the following notation: $N$ is the number of terms of the truncated spectral expansion, $T_{\text{sp.exp.}}$ is the time required to compute the spectrum $\{\tilde\lambda_n\}_{1\leq n\leq N}$ and evaluate the functions $\tilde\phi_{n,\alpha}^{\ell,-}(x)$ in (\ref{spec_exp}), and $T_{\text{quad.}}$ is the time required to numerically evaluate the integrals in (\ref{CallValue})--(\ref{PutValue}). Note that the computational time for the Bessel-K model is much smaller than in the other models thanks to fast and robust numerical Matlab routines for computing Bessel functions. Computations could be drastically sped up if the Matlab routines were translated into the machine code, faster and more robust routines for computing confluent hypergeometric functions were available, and the code was further optimized.

\noindent\textbf{Test 2.}\\
We note that the spectral expansion algorithm allows us to efficiently and simultaneously compute option values and deltas for several different strikes and spots. Thus for the second numerical test we calculated option values and deltas for a range of spot values. To speed up the computations, we used the Simpson quadrature rule with a uniform grid. It allowed us to parallelize the computations of spectral expansions by computing individually the parts that only depend on $y$, $x$, and $\ell$, respectively. The model parameters and problem parameters used here were the same as those for the first test. Figures~\ref{fig1}--\ref{fig4} demonstrate the results obtained.

\noindent\textbf{Test 3.}\\
The results obtained for the CEV model are compared with the Monte Carlo (biased) estimates with $M=10^6$ sample paths and $\Delta t=0.05$. The results of the numerical tests are presented in Table~\ref{table2}. We observe good agreement between the results provided by the spectral expansion method and the Monte Carlo algorithm.

\noindent\textbf{Test 4.}\\
In another numerical test, we study the sensitivity of the step option price under the Bessel-K model as the local
volatility function changes its steepness. The steepness of the local volatility was controlled by varying the
parameter $\mu$ from 0.1 to 0.9. The parameter $c$ was calibrated so that the local volatility at $S_0=100$ is fixed to 25\% in all cases. Figure~\ref{fig:BKmu} contains both plots of local volatility functions and the step-down put option prices. We observe that an increase in the steepness of the local volatility tends to decrease the step-down put values. This seems rather intuitive as an increase in steepness tends to increase the occupation time below a~given level $L$.

\noindent\textbf{Test 5.}\\
Step option prices converge to a~barrier option price as $\alpha\to\infty$. In fact, the Green's functions, and hence the respective transition PDFs converge to the respective functions for the process having the given upper or lower killing barrier level $L$. In the next numerical example (Figure~\ref{fig:BKrho} and Table~\ref{table3}) we show, by computational implementation of the spectral expansions for the transition PDF, how the price and delta sensitivity of a~step-down call option changes as $\alpha$ increases. The computations were done for the Bessel-K model using the parameters in Table~\ref{tb:Models}. The spot price and strike are fixed at 100.

\noindent\textbf{Test 6.}\\
We studied the convergence of the spectral expansion method as $N$, the number of terms, increases. Figure~\ref{fig:convA} illustrates the convergence of the PDF $\tilde{p}_\alpha^{\ell,-}$ given by (\ref{spec_expN}) as $N$ increases. The computations were done for the Bessel-K model whose parameters are specified in Table~\ref{tb:Models}. The accompanying table in Figure~\ref{fig:convB} contains the step-down call and put option prices for $S_0=K=100$, corresponding to using the truncated spectral expansion in (\ref{spec_expN}) for relatively small $N$ number of terms. We hence observe typical rapid convergence in the computed prices with the use of the truncated spectral expansion formula.

\noindent\textbf{Test 7.}\\
Recall that the eigenvalues $\{\tilde\lambda_n\}_{n\geq 1}$ grow linearly (for the CEV, confluent-U, and UOU models) or quadratically (for the Bessel-K model with killing) as $n$ increases. Therefore, the introduction of a killing upper barrier in a hypergeometric diffusion model allows us to accelerate the convergence of spectral series expansions. Since the step options are here defined such that they become worthless if the upper level $\h $ is hit before the maturity time, the option values are biased. In this last numerical test we study how such a bias depends on the level $\h $. Figure~\ref{fig:levelA} contains the graph of the initial price $C^-_{\mathrm{step}}(S_0=100,T=0.5,K=100)$ of a step-down call option plotted as a function of the level $\h $. As $\h \nearrow\infty$, the probability of hitting the level $\h $ decreases and the option price increases, asymptotically approaching the option price for the model without killing at an upper level. Again, the computations were done for the Bessel-K model. Table~\ref{fig:levelB} gives the put and call option values when
the upper killing level $\h $ changes from 150 to 400.

\section{Conclusions} \label{sec:7}
One main contribution of this paper is the development of new analytically closed-form spectral expansion formulae for the transition probability
density function under the CEV model, and under various other solvable families of multi-parameter diffusion processes having nonlinear local volatility,
in the presence of killing at an exponential stopping time (independent of the process) of occupation above or below any fixed level.
The spectral expansions in Proposition~\ref{propn_spectral} are applicable to a~general class of diffusions. This paper has successfully implemented the spectral
expansions under the CEV model and three other main families of nonlinear local volatility models. As shown in recent papers, the nonlinear local volatility models are useful for describing asset price dynamics and for pricing standard, lookback and barrier options in finance. This paper further succeeds in providing an analytical framework for the risk-neutral pricing of classes of occupation-time options under these models.
In particular, numerical test results show that the spectral expansions are rapidly convergent and provide an efficient method to compute the prices of any proportional step-up and step-down options. Moreover, the option Greeks (e.g. delta sensitivity) are also readily and simultaneously calculated by simply taking analytical derivatives of the spectral expansions without any loss of precision. The spectral expansions converge more rapidly with increasing time and hence the prices and Greeks are computed even more efficiently for longer dated options. This offers a~significant computational advantage in comparison to any Monte Carlo method.

This paper also derives a~general analytical expression for the resolvent kernel (i.e. Green's function) of solvable diffusions with killing at an exponential stopping time. This result, by itself, is also useful for analytically computing the Laplace transform of certain conditional expectations involving functionals of the occupation time for various families of diffusion process above or below a~given level. In particular, Lemma \ref{Lemma_Greenfunc} gives analytical formulae for the conditional
expectations in equations (\ref{expectations-for-stopping1}) and (\ref{expectations-for-stopping2}) for any solvable diffusion model. For example, these formulae automatically generate the expressions tabulated in \cite{BS02} for various drifted Brownian motions, geometric Brownian motion, the (squared) Bessel process, the Ornstein-Uhlenbeck (OU) and radial OU processes. Moreover, by the Doob transform employed in this paper, the respective expressions for the Laplace transform of the conditional expectations now also extend readily to various newly solvable diffusions. The results follow simply from the known fundamental solutions for the underlying process and their Wronskians.

The theoretical development in this paper also sets the foundation for further analytical extensions and applications involving the occupation time of newly solvable diffusion processes. For example, Lemma \ref{Lemma_Greenfunc} can also be extended to cover expectation formula for the case of killing in proportion to the occupation time between two levels or a~linear combination of occupation times below and above a~ given level. In turn, Proposition~\ref{propn_spectral} can then be extended to cover such cases. As long as we are in Spectral Category I, the spectral expansion formulae for the relevant transition probability density functions will have a~series representation. The inclusion of additionally imposed killing, i.e. the usual restrictions on the supremum or infimum of the process by specifying one or two interior levels, is also readily handled via the fundamental solutions with appropriately posed boundary conditions.

\appendix

\section{Proofs} \label{sect_A1}
\subsection{The proof of Lemma \ref{Lemma_Greenfunc}} \label{subsect_A11}
The expectations in (\ref{expectations-for-stopping1}) and (\ref{expectations-for-stopping2}) are respectively given by
$$\E_x\left[ \e^{-\alpha A_{\tau,\mathbf{X}}^{\ell,\pm}}\,;\,X_\tau \in \d y\right]
\equiv {\partial \over \partial y} \E_x\big[ \e^{-\alpha A_{\tau,\mathbf{X}}^{\ell,\pm}}\indfunc_{X_\tau < y}\big]\,\d y
= \lambda\widetilde{G}_\alpha^{\ell,\pm}(x,y,\lambda)\,\d y.$$
Hence, by a~standard application of the Feynman-Kac formula [e.g. see pages 105-106 in \cite{BS02}, but here generalized to the diffusion with generator defined in (\ref{Generator_X}) above] the respective Green's functions satisfy the ordinary differential equations
$$(\G - (\lambda + \alpha\indfunc_{x\le \ell}))\widetilde{G}_\alpha^{\ell,-}(x,y,\lambda) = -\delta (x - y) \,\text{ and } \,
(\G - (\lambda + \alpha\indfunc_{x\ge \ell}))\widetilde{G}_\alpha^{\ell,+}(x,y,\lambda) = -\delta (x - y),$$
where $\delta(\cdot)$ is the Dirac delta function. These two equations are solved in the same manner as follows.

Consider the equation in $\widetilde{G}_\alpha^{\ell,-}$. For $x \le \ell$, the solution is a~linear combination of the
pair of fundamental solutions $\{\phi_{\lambda + \alpha},\psi_{\lambda + \alpha}\}$ of $(\G - (\lambda + \alpha))\varphi = 0$.
For $x > \ell$, the solution is a~linear combination of the pair of fundamental solutions $\{\phi_\lambda,\psi_\lambda\}$ of
$(\G - \lambda)\varphi = 0$. A pair $\{\tilde\psi_\lambda,\tilde\phi_\lambda\}$ of fundamental solutions
for the equation $(\G - (\lambda + \alpha\indfunc_{x\le \ell}))\widetilde{G}_\alpha^{\ell,-} = 0$,
with the same respective left and right boundary conditions as the pair $\{\psi_\lambda,\phi_\lambda\}$, is hence given by
\begin{equation}  \label{tilde_funcs}
\tilde{\psi}_\lambda(x) = \left\{\begin{array}{ll}
      \psi_{\lambda + \alpha}(x), & x \le \ell,\\[10pt]
      A\psi_\lambda(x) + B\phi_\lambda(x), & x > \ell,
   \end{array}\right.
   \quad
   \tilde{\phi}_\lambda(x) = \left\{\begin{array}{ll}
      C\psi_{\lambda + \alpha}(x) + D\phi_{\lambda + \alpha}(x), & x \le \ell,\\[10pt]
      \phi_\lambda(x), & x > \ell.
   \end{array}\right.
\end{equation}
The constants $A,B,C,D$ are uniquely determined by requiring that these functions are in $\mathcal{C}^1(\I)$, i.e.
at $x=\ell$: $\tilde{\psi}_\lambda(\ell -) = \tilde{\psi}_\lambda(\ell +)$,
$\tilde{\psi}_\lambda'(\ell -) = \tilde{\psi}_\lambda'(\ell +)$ and
$\tilde{\phi}_\lambda(\ell -) = \tilde{\phi}_\lambda(\ell +)$,
$\tilde{\phi}_\lambda'(\ell -) = \tilde{\phi}_\lambda'(\ell +)$. The first set of conditions is a~linear system of equations in $A,B$
with solution $A=\frac{\wron{\phi}{\psi}{\lambda}{\lambda + \alpha}(\ell)}{\wron{\phi}{\psi}{\lambda}{\lambda}(\ell)}$,
$B=\frac{\wron{\psi}{\psi}{\lambda + \alpha}{\lambda}(\ell)}{\wron{\phi}{\psi}{\lambda}{\lambda}(\ell)}$
and the second set of conditions is a~linear system in $C,D$ with solution
$C=\frac{\wron{\phi}{\phi}{\lambda + \alpha}{\lambda}(\ell)}{\wron{\phi}{\psi}{\lambda + \alpha}{\lambda + \alpha}(\ell)}$,
$D=\frac{\wron{\phi}{\psi}{\lambda}{\lambda + \alpha}(\ell)}{\wron{\phi}{\psi}{\lambda + \alpha}{\lambda + \alpha}(\ell)}$.
Computing the Wronskian for the pair in (\ref{tilde_funcs}) gives
\begin{eqnarray*}
W[\tilde{\phi}_\lambda, \tilde{\psi}_\lambda](x) = \left\{\begin{array}{ll}
      \frac{\wron{\phi}{\psi}{\lambda + \alpha}{\lambda + \alpha}(x)}{\wron{\phi}{\psi}{\lambda + \alpha}
      {\lambda + \alpha}(\ell)}\wron{\phi}{\psi}{\lambda}{\lambda + \alpha}(\ell), & x \le \ell,\\[10pt]
      \frac{\wron{\phi}{\psi}{\lambda}{\lambda}(x)}{\wron{\phi}{\psi}{\lambda}
      {\lambda}(\ell)}\wron{\phi}{\psi}{\lambda}{\lambda + \alpha}(\ell), & x > \ell
   \end{array}\right. = \frac{\wron{\phi}{\psi}{\lambda}{\lambda + \alpha}(\ell)}{\s(\ell)} \s(x) \equiv \widetilde{\mathcal W}_\lambda \, \s(x).
\end{eqnarray*}
Here we used the identity $\wron{\phi}{\psi}{\gamma}{\gamma}(x) = {\mathcal W}_\gamma \,\s(x)$ so that
$\wron{\phi}{\psi}{\gamma}{\gamma}(x) / \wron{\phi}{\psi}{\gamma}{\gamma}(\ell) = \s(x) / \s(\ell)$ for any $\gamma \in\C$.
The Green's function $\widetilde{G}_\alpha^{\ell,-}$ is then simply given by combining this Wronskian with (\ref{tilde_funcs}):
\begin{equation}
{\widetilde{G}_\alpha^{\ell,-}(x,y,\lambda) \over \m(y)}
= \frac{\tilde\psi_\lambda(x \wedge y)\tilde\phi_\lambda(x\vee y)}{\widetilde{\mathcal W}_\lambda}
= \s(\ell)\frac{\tilde\psi_\lambda(x \wedge y)\tilde\phi_\lambda(x\vee y)}{\wron{\phi}{\psi}{\lambda}{\lambda + \alpha}(\ell)}
\label{greenfunc_tilde1}
\end{equation}
for $x,y\in\I$. The Green's function in (\ref{greenfunc_tilde1}) can be recast as in equation (\ref{Greenfunc_minus}) by using
(\ref{tilde_funcs}), and the definition for $G$ in (\ref{greenfunc}), for the four different cases: $x,y\le\ell$, $x\le \ell \le y$, $x\ge \ell \ge y$
or $x,y\ge\ell$.

The derivation for $\widetilde{G}_\alpha^{\ell,+}$ in equation (\ref{Greenfunc_plus}) follows the same steps as above. A general solution to the equation
$(\G - (\lambda + \alpha\indfunc_{x\ge \ell}))\widetilde{G}_\alpha^{\ell,+} = 0$
is a~linear combination of the pair $\{\psi_\lambda,\phi_\lambda\}$, for $x < \ell$, and of the pair
$\{\phi_{\lambda + \alpha},\psi_{\lambda + \alpha}\}$ for $x \ge \ell$.
In analogy with equation (\ref{tilde_funcs}), a~pair of solutions $\{\tilde\psi_\lambda,\tilde\phi_\lambda\}$,
with the same left and right boundary conditions as the pair $\{\psi_\lambda,\phi_\lambda\}$ is as follows:
\begin{equation}  \label{tilde_funcs2}
\tilde{\psi}_\lambda(x) = \left\{\begin{array}{ll}
      \psi_\lambda(x), & x < \ell,\\[10pt]
      A\psi_{\lambda + \alpha}(x) + B\phi_{\lambda + \alpha}(x), & x \ge \ell,
   \end{array}\right.
   \quad
   \tilde{\phi}_\lambda(x) = \left\{\begin{array}{ll}
      C\psi_\lambda(x) + D\phi_\lambda(x), & x < \ell,\\[10pt]
      \phi_{\lambda + \alpha}(x), & x \ge \ell.
   \end{array}\right.
\end{equation}
The requirement that these functions are in $\mathcal{C}^1(\I)$ then uniquely gives
$A=\frac{\wron{\phi}{\psi}{\lambda + \alpha}{\lambda}(\ell)}{\wron{\phi}{\psi}{\lambda + \alpha}{\lambda + \alpha}(\ell)}$,
$B=\frac{\wron{\psi}{\psi}{\lambda}{\lambda + \alpha}(\ell)}{\wron{\phi}{\psi}{\lambda + \alpha}{\lambda + \alpha}(\ell)}$,
$C=\frac{\wron{\phi}{\phi}{\lambda}{\lambda + \alpha}(\ell)}{\wron{\phi}{\psi}{\lambda}{\lambda}(\ell)}$,
$D=\frac{\wron{\phi}{\psi}{\lambda + \alpha}{\lambda}(\ell)}{\wron{\phi}{\psi}{\lambda}{\lambda}(\ell)}$. The Wronskian of the two solutions
in (\ref{tilde_funcs2}) is readily computed to be
\begin{eqnarray*}
W[\tilde{\phi}_\lambda, \tilde{\psi}_\lambda](x)
 = \frac{\wron{\phi}{\psi}{\lambda + \alpha}{\lambda}(\ell)}{\s(\ell)} \s(x) \equiv \widetilde{\mathcal W}_\lambda \, \s(x).
\end{eqnarray*}
Combining this Wronskian with (\ref{tilde_funcs2}) gives the Green's function:
\begin{equation}
{\widetilde{G}_\alpha^{\ell,+}(x,y,\lambda) \over \m(y)}
= \frac{\tilde\psi_\lambda(x \wedge y)\tilde\phi_\lambda(x\vee y)}{\widetilde{\mathcal W}_\lambda}
= \s(\ell)\frac{\tilde\psi_\lambda(x \wedge y)\tilde\phi_\lambda(x\vee y)}{\wron{\phi}{\psi}{\lambda + \alpha}{\lambda}(\ell)}
\label{greenfunc_tilde2}
\end{equation}
for $x,y\in\I$. The Green's function takes the more explicit form in (\ref{Greenfunc_plus}) by using
(\ref{tilde_funcs2}), and the definition for $G$ in (\ref{greenfunc}), for the four different cases: $x,y\le\ell$, $x\le \ell \le y$, $x\ge \ell \ge y$
or $x,y\ge\ell$. This completes the proof.
\vskip 0.1in
[{\it Remark}: In the limit $\alpha\to\infty$, the Green's functions $\widetilde{G}_\alpha^{\ell,\pm}$
can be proven to converge to the respective Green's functions for the process with killing at an upper (or lower) level $\ell$.
We do not give a~proof here, as it is based on the $\alpha\to\infty$ formal asymptotic analysis of the fundamental solutions
$\psi_{\lambda + \alpha}(x)$ and $\phi_{\lambda + \alpha}(x)$. By the leading term asymptotics,
$\psi_{\lambda + \alpha}(x)/ \psi_{\lambda + \alpha}'(x)\to 0$  and $\phi_{\lambda + \alpha}(x)/ \phi_{\lambda + \alpha}'(x)\to 0$.
Hence, for example,
$\frac{\wron{\phi}{\phi}{\lambda}{\lambda+\alpha}(\ell)}{\wron{\phi}{\psi}{\lambda+\alpha}{\lambda}(\ell)}\to -\frac{\phi_\lambda(\ell)}{\psi_\lambda(\ell)}$ and
$\frac{\wron{\psi}{\psi}{\lambda+\alpha}{\lambda}(\ell)}{\wron{\phi}{\psi}{\lambda}{\lambda+\alpha}(\ell)}\to -\frac{\psi_\lambda(\ell)}{\phi_\lambda(\ell)}$.
Using these limits and the asymptotic properties we can arrive at the asymptotic forms for the Green's functions in equations
(\ref{Greenfunc_plus}) and (\ref{Greenfunc_minus}) of Lemma \ref{Lemma_Greenfunc}. In particular, as $\alpha\to\infty$:
\begin{equation*}
{\widetilde{G}_\alpha^{\ell,+}(x,y,\lambda) \over \m(y)} \to {G^{\ell,+}(x,y,\lambda) \over \m(y)} \equiv
\frac{\psi_\lambda(x \wedge y){\mathcal S}(x\vee y,\ell;\lambda)}{{\mathcal W}_\lambda \psi_\lambda(\ell)},\,\,\,\,x,y \le \ell,
\end{equation*}
with $G^{\ell,+}(x,y,\lambda)\equiv 0$ if $x > \ell$ or $y > \ell$, and
\begin{equation*}
{\widetilde{G}_\alpha^{\ell,-}(x,y,\lambda) \over \m(y)} \to {G^{\ell,-}(x,y,\lambda) \over \m(y)} \equiv
\frac{{\mathcal S}(\ell,x\wedge y;\lambda)\phi_\lambda(x \vee y)}{{\mathcal W}_\lambda \phi_\lambda(\ell)},\,\,\,\,x,y \ge \ell,
\end{equation*}
with $G^{\ell,-}(x,y,\lambda)\equiv 0$ if $x < \ell$ or $y < \ell$. Here, $G^{\ell,+}(x,y,\lambda)$, or $G^{\ell,-}(x,y,\lambda)$, are the
respective Green's function for the process $\mathrm{X} < \ell$, or $\mathrm{X} > \ell$, with killing imposed at the upper, or lower,
level $\ell$. The generalized cylinder function is defined by
$$
{\mathcal S}(x,y;\lambda) \deq \phi_\lambda(x)\psi_\lambda(y) - \psi_\lambda(x)\phi_\lambda(y).
$$
For real $\lambda >0$, ${\mathcal S}(x,\ell;\lambda)$ (${\mathcal S}(\ell,x;\lambda)$) is a~decreasing (increasing) positive function
for $x\le \ell$ ($x \ge \ell$).]

\subsection{The proof of Proposition \ref{propn_spectral}} \label{subsect_A12}

Due to the similar structure of the Green's functions $\widetilde{G}_\alpha^{\ell,+}$ and $\widetilde{G}_\alpha^{\ell,-}$, as observed in equations
(\ref{Greenfunc_plus}) and (\ref{Greenfunc_minus}), we will only present the proof for $\tilde{p}_\alpha^{\ell,+}$,
i.e. equations (\ref{spec_exp}) and (\ref{spec_exp_a}). The PDF is given by the Laplace inverse
$\tilde{p}_\alpha^{\ell,+}(t;x,y) = {\mathcal L}_\lambda^{-1}\left(\widetilde{G}_\alpha^{\ell,+}(x,y,\lambda)\right)(t)$
which can be computed for all four separate cases in equation (\ref{Greenfunc_plus})
with standard use of the Residue Theorem upon closing the Bromwich contour integral on the left-half of the complex $\lambda$ plane.
In particular,
$${\tilde{p}_\alpha^{\ell,+}(t;x,y)\over \m(y)} = \sum_{n=1}^\infty \e^{-\tilde\lambda_n t}
\,\text{Res}\left[ {\widetilde{G}_\alpha^{\ell,+}(x,y,\lambda)\over \m(y)} \,;\, \lambda = -\tilde{\lambda}_n\right]$$
where $\text{Res}\left[ {\widetilde{G}_\alpha^{\ell,+}(x,y,\lambda)\over \m(y)} \,;\, \lambda = -\tilde{\lambda}_n\right] =
\tilde\phi_{n,\alpha}^{\ell,+}(x)\tilde\phi_{n,\alpha}^{\ell,+}(y)$.

Consider the case where $x\le \ell \le y$, i.e.
${\widetilde{G}_\alpha^{\ell,+}(x,y,\lambda)\over \m(y)} = {\s(\ell)\over \wron{\phi}{\psi}{\lambda+\alpha}{\lambda}(\ell)}
\psi_\lambda(x)\phi_{\lambda+\alpha}(y)$. The only singularities are the real simple zeros at $\lambda = -\tilde{\lambda}_n, n\ge 1,$ of the Wronskian
$\wron{\phi}{\psi}{\lambda+\alpha}{\lambda}(\ell)$, where $\{\tilde\lambda_n\}_{n\ge 1}$ is an increasing sequence of eigenvalues solving
$\wron{\phi}{\,\psi}{-\tilde\lambda_n +\alpha}{-\tilde\lambda_n}(\ell) = 0$.
By analyticity of the fundamental functions (w.r.t. $\lambda$), the residue of the Green's function at these simple poles is then given by:
$$\text{Res}\left[{\widetilde{G}_\alpha^{\ell,+}(x,y,\lambda)\over \m(y)}\,;\,\lambda = -\tilde{\lambda}_n \right]
= \s(\ell)\,\text{Res}\left[{\psi_\lambda(x)\phi_{\lambda+\alpha}(y)\over \wron{\phi}{\psi}{\lambda+\alpha}{\lambda}(\ell)}
\,;\,\lambda = -\tilde{\lambda}_n \right] =
\s(\ell)\frac{\psi_{\!-\tilde\lambda_n}(x)
\phi_{\!-\tilde\lambda_n +\alpha}(y)}{{d\over d\lambda}\wron{\phi}{\,\psi}{\lambda +\alpha}{\lambda}(\ell)\vert_{\lambda=-\tilde\lambda_n}}.
$$
The case where $y\le \ell \le x$ is similar, where
$$\text{Res}\left[{\widetilde{G}_\alpha^{\ell,+}(x,y,\lambda)\over \m(y)}\,;\,\lambda = -\tilde{\lambda}_n \right]
= \s(\ell)\,\text{Res}\left[{\phi_{\lambda+\alpha}(x)\psi_\lambda(y)\over \wron{\phi}{\psi}{\lambda+\alpha}{\lambda}(\ell)}
\,;\,\lambda = -\tilde{\lambda}_n \right] =
\s(\ell)\frac{\phi_{\!-\tilde\lambda_n +\alpha}(x)\psi_{\!-\tilde\lambda_n}(y)
}{{d\over d\lambda}\wron{\phi}{\,\psi}{\lambda +\alpha}{\lambda}(\ell)\vert_{\lambda=-\tilde\lambda_n}}.
$$

We now consider the case where $x\le \ell, y \le \ell$, i.e.
$${\widetilde{G}_\alpha^{\ell,+}(x,y,\lambda)\over \m(y)} =
{G(x,y,\lambda)\over \m(y)} + \displaystyle\frac{\wron{\phi}{\phi}{\lambda}{\lambda+\alpha}(\ell)}{\wron{\phi}{\psi}{\lambda+\alpha}{\lambda}(\ell)}
{\psi_\lambda(x)\psi_\lambda(y) \over {\mathcal W}_\lambda},$$
with ${G(x,y,\lambda)\over \m(y)}$ given by equation (\ref{greenfunc}). Hence, this Green's function has {\it two sets of singularities}.
One is the set of simple poles corresponding to the simple zeros solving ${\mathcal W}_{\lambda = -\lambda_n} = 0$, i.e. $\{\lambda_n\}_{n\ge 1}$
denotes the eigenvalue set for the Sturm-Liouville problem with generator $\G$ for diffusion $\mathrm{X}$ on $\I$.
The other is the set of zeros $\lambda = -\tilde{\lambda}_n, n\ge 1,$ of $\wron{\phi}{\psi}{\lambda+\alpha}{\lambda}(\ell)$.
We now establish that the only nonzero residues are for the set $\lambda = -\tilde{\lambda}_n, n\ge 1$. Assume the set $\{\lambda_n\}_{n\ge 1}$
is isolated from the set $\{\tilde\lambda_n\}_{n\ge 1}$. Then, computing the residue at every simple pole $\lambda = -\lambda_n$ gives:
\begin{align*}
\text{Res}\left[{\widetilde{G}_\alpha^{\ell,+}(x,y,\lambda)\over \m(y)};\lambda = -\lambda_n \right] &=
\text{Res}\left[{G(x,y,\lambda)\over \m(y)};\lambda = -\lambda_n \right] +
\text{Res}\left[\displaystyle\frac{\wron{\phi}{\phi}{\lambda}{\lambda+\alpha}(\ell)}{\wron{\phi}{\psi}{\lambda+\alpha}{\lambda}(\ell)}
{\psi_\lambda(x)\psi_\lambda(y) \over {\mathcal W}_\lambda};\lambda = -\lambda_n \right].
\end{align*}
For $\lambda = -\lambda_n$, the fundamental functions $\psi_\lambda$ and $\phi_\lambda$ are proportional to each other, i.e.
$\phi_{-\lambda_n}(x) = A_n\psi_{-\lambda_n}(x)$, for some constant $A_n\ne 0$. Hence, the ratio of Wronskians in the above second residue term evaluates to
$\frac{\wron{\phi}{\phi}{\lambda}{\lambda+\alpha}(\ell)}{\wron{\phi}{\psi}{\lambda+\alpha}{\lambda}(\ell)}\bigg\vert_{\lambda = -\lambda_n} = -A_n$.
Denoting $C_n \deq {d\over d\lambda} {\mathcal W}_\lambda\big\vert_{\lambda = -\lambda_n}$, the second residue evaluates to
\begin{align*}
\text{Res}\left[\displaystyle\frac{\wron{\phi}{\phi}{\lambda}{\lambda+\alpha}(\ell)}{\wron{\phi}{\psi}{\lambda+\alpha}{\lambda}(\ell)}
{\psi_\lambda(x)\psi_\lambda(y) \over {\mathcal W}_\lambda};\lambda = -\lambda_n \right]
= -{A_n\over C_n}\psi_{-\lambda_n}(x)\psi_{-\lambda_n}(y) = - \phi_n(x)\phi_n(y)
\end{align*}
where $\phi_n(x) \deq \pm\sqrt{{A_n\over C_n}}\psi_{-\lambda_n}(x)$ is the $n$-th eigenfunction of $-\G$ for $x\in\I$.
The first residue has the standard eigenfunction product form:
\begin{align*}
\text{Res}\left[{G(x,y,\lambda)\over \m(y)};\lambda = -\lambda_n \right] = \phi_n(x)\phi_n(y).
\end{align*}
Adding the two terms gives a~zero residue at every $\lambda = -\lambda_n$, i.e.
$\text{Res}\left[{\widetilde{G}_\alpha^{\ell,+}(x,y,\lambda)\over \m(y)};\lambda = -\lambda_n \right] = 0$.

The only nonzero residues are hence due to the assumed simple poles $\lambda = -\tilde{\lambda}_n, n\ge 1$ and these are
\begin{align*}
\text{Res}\left[{\widetilde{G}_\alpha^{\ell,+}(x,y,\lambda)\over \m(y)};\lambda = -\tilde\lambda_n \right] &=
\text{Res}\left[\displaystyle\frac{\wron{\phi}{\phi}{\lambda}{\lambda+\alpha}(\ell)}{\wron{\phi}{\psi}{\lambda+\alpha}{\lambda}(\ell)}
{\psi_\lambda(x)\psi_\lambda(y) \over {\mathcal W}_\lambda};\lambda = -\tilde\lambda_n \right]
\\
&=\left[\frac{\wron{\phi}{\phi}{\lambda}{\lambda +\alpha}(\ell)}
{{\mathcal W}_{\lambda}\,{d\over d\lambda}\wron{\phi}{\,\psi}{\lambda +\alpha}{\lambda}(\ell)}\right]_{\lambda=-\tilde\lambda_n}
\psi_{\!-\tilde\lambda_n}(x)\psi_{\!-\tilde\lambda_n}(y)
\end{align*}
which is the form in equation (\ref{spec_exp_a}) for $x\le \ell, y \le \ell$. We note that if the $n$-th eigenvalue $\tilde\lambda_n$
happens to also coincide with an eigenvalue in the set $\{\lambda_n\}_{n\ge 1}$, say $\tilde\lambda_n = \lambda_m$ for some $m\ge 1$,
then the above formula is interpreted as a~limit $\lambda \to -\tilde\lambda_n$.
We remark that for a~coalescence of zeros (w.r.t. $\lambda$) of both Wronskians,
${\mathcal W}_{\lambda}$ and $\wron{\phi}{\,\psi}{\lambda +\alpha}{\lambda}(\ell)$, the point $\lambda = -\tilde\lambda_n$
is still a~first order pole since $\phi_\lambda = A_n\psi_\lambda$, and hence the Wronskian $\wron{\phi}{\phi}{\lambda}{\lambda +\alpha}(\ell)$
in the numerator is proportional to $\wron{\phi}{\,\psi}{\lambda +\alpha}{\lambda}(\ell)$ in the denominator, at $\lambda = -\tilde\lambda_n$.

The last case where $x\ge \ell, y \ge \ell$ follows in very similar fashion. Again, the residues for
the set $\lambda = -\lambda_n, n\ge 1,$ are all zero and the only nonzero residues are due to
the assumed simple poles $\lambda = -\tilde{\lambda}_n, n\ge 1,$ where
\begin{align*}
\text{Res}\left[{\widetilde{G}_\alpha^{\ell,+}(x,y,\lambda)\over \m(y)};\lambda = -\tilde\lambda_n \right] &=
\text{Res}\left[\displaystyle\frac{\wron{\psi}{\psi}{\lambda}{\lambda+\alpha}(\ell)}{\wron{\phi}{\psi}{\lambda+\alpha}{\lambda}(\ell)}
  {\phi_{\lambda+\alpha}(x)\phi_{\lambda+\alpha}(y) \over {\mathcal W}_{\lambda+\alpha}};\lambda = -\tilde\lambda_n \right]
\\
&=\left[\frac{\wron{\psi}{\psi}{\lambda}{\lambda +\alpha}(\ell)}
{{\mathcal W}_{\lambda+\alpha}\,{d\over d\lambda}\wron{\phi}{\,\psi}{\lambda +\alpha}{\lambda}(\ell)}\right]_{\lambda=-\tilde\lambda_n}
\phi_{\!-\tilde\lambda_n+\alpha}(x)\phi_{\!-\tilde\lambda_n+\alpha}(y).
\end{align*}
Again, the same above remarks apply here as for the previous expression just above.

\bibliographystyle{plain} 
\bibliography{rmakarov-bib}

\begin{thebibliography}{10}

\bibitem{AbrSteg}
M.~Abramowitz and I.A. Stegun.
\newblock {\em Handbook of Mathematical Functions}.
\newblock New York: Dover, 1972.

\bibitem{BS02}
A.~N. Borodin and P.~Salminen.
\newblock {\em Handbook of Brownian Motion -- Facts and Formulae}.
\newblock Probability and its Applications. Birkh\"{a}user Basel, 2 edition,
  2002.

\bibitem{CCW10}
N.~Cai, N.~Chen, and X.~Wan.
\newblock Occupation times of jump-diffusion processes with double exponential
  jumps and the pricing of options.
\newblock {\em Mathematics of Operations Research}, 35(2):412--437, 2010.

\bibitem{CM08}
G.~Campolieti and R.~N. Makarov.
\newblock {M}onte {C}arlo path integral pricing of {A}sian options on state
  dependent volatility models using high performance computing.
\newblock {\em Quantitative Finance}, 8(2):147--161, 2008.

\bibitem{CM11}
G.~Campolieti and R.~N. Makarov.
\newblock Dual stochastic transformations of solvable diffusions.
\newblock {\em Stochastics: An International Journal of Probability and
  Stochastic Processes}, 2012.
\newblock Submitted.

\bibitem{CM10}
G.~Campolieti and R.~N. Makarov.
\newblock On properties of analytically solvable families of local volatility
  diffusion models.
\newblock {\em Mathematical Finance}, 22(3):488--518, 2012.

\bibitem{Cox75}
J.C. Cox.
\newblock Notes on option pricing {I}: Constant elasticity of variance
  diffusions.
\newblock {\em Journal of Portfolio Management}, 22:15--17, 1996.
\newblock Published first as a working paper, {S}tanford {U}niversity, 1975.

\bibitem{Dass95}
A.~Dassios.
\newblock The distribution of the quantile of a {Brownian} motion with drift
  and the pricing of related path-dependent options.
\newblock {\em Annals of Applied Probability}, pages 389--398, 1995.

\bibitem{FT01}
G.~Fusai and A.~Tagliani.
\newblock Pricing of occupation time derivatives: continuous and discrete
  monitoring.
\newblock {\em Journal of Computational Finance}, 5:1--37, 2001.

\bibitem{Hugo99}
J.~Hugonnier.
\newblock The {Feynman-Kac} formula and pricing ocupation time derivatives.
\newblock {\em International Journal of Theoretical and Applied Finance},
  2(153--178):1999, 2.

\bibitem{JR98}
J.~C. Jackwerth and M.~Rubinstein.
\newblock Recovering stochastic processes from option prices.
\newblock {\em Working paper, University of California, Berkeley, CA}, 1998.

\bibitem{LK07}
S.~L. Leung and Y.~K. Kwok.
\newblock Distribution of occupation times for constant elasticity of variance
  diffusion and pricing of the $\alpha$-quantile options.
\newblock {\em Quantitative Finance}, 7(1):87--94, 2007.

\bibitem{Linet99}
V.~Linetsky.
\newblock Step options.
\newblock {\em Mathematical Finance}, 9(1):55--96, 1999.

\bibitem{Linet02}
V.~Linetsky.
\newblock Structuring, pricing and hedging double-barrier step options.
\newblock {\em Journal of Computational Finance}, 5(2):55--87, 2002.

\bibitem{Linet04}
V.~Linetsky.
\newblock The spectral decomposition of the option value.
\newblock {\em International Journal of Theoretical and Applied Finance},
  7:337--384, 2004.

\bibitem{MW11}
R.~Makarov and K.~Wouterloot.
\newblock Exact simulation of occupation times.
\newblock In {Wozniakowski, H.} and {Plaskota, L.}, editors, {\em Monte Carlo
  and Quasi-Monte Carlo Methods 2010}, pages 575--589. Springer-Verlag Berlin
  Heidelberg, 2012.
\newblock In press.

\bibitem{MG10}
R.~N. Makarov and D.~Glew.
\newblock Exact simulation of {Bessel} diffusions.
\newblock {\em Monte Carlo Methods and Applications}, 16(3):283--306, 2010.

\end{thebibliography}

\newpage

\begin{table}
  \centering
  \caption{Parameters of the four asset price models}\label{tb:Models}
  \begin{tabular}{ll}
    \textbf{Model} & \textbf{Parameters} \\
    \hline
    CEV & $\delta=2500$; $\beta=-2$; $\mathrm{r}=0.02$\\
    \hline
    Bessel-K (BK) & $\mu = 0.5$; $\gamma_0 = 2.2$; $\rho = 0.00001$;
      $c = 728.7467627$; $\h =500$; $\mathrm{r}=0.02$ \\
    \hline
    Confluent-U (CU) & $\mu = 0.5$; $c = 133.1173736$; $\rho = 0.01$; $\nu = \sqrt{2}$; $\gamma_1 = 0.1$; $\mathrm{r}=0.02$ \\
    \hline
    UOU & $\rho = 0.001$; $\nu = 2$; $c = 71.11606167$; $\gamma_1 = 0.2$; $\mathrm{r}=0.02$\\
    \hline
  \end{tabular}
\end{table}

\begin{table}
  \centering
  \setlength{\tabcolsep}{10pt}
  \caption{Values of the step-down call and put options computed for a~range of strikes under the four asset price models. The parameters used are $S_0 = 100$, $T=0.5$, $\alpha = 5$, $L=90$.}
  \label{table1A}
    \begin{tabular}{rr@{\;}r@{\;}r@{\;}rr@{\;}r@{\;}r@{\;}r}
      \hline\noalign{\smallskip}
      &\multicolumn{4}{c}{Step Calls}&\multicolumn{4}{c}{Step Puts}\\
      $K$ & CEV & BK & CU & UOU & CEV & BK & CU & UOU\\
      \hline\noalign{\smallskip}
80  &    20.364424  &  19.774476  &  20.657331  & 20.603226  &   0.295586   &   0.167632  &   0.340048  &     0.013590  \\
90  &    13.359199  &  12.953074  &  13.563118  & 13.043498  &   0.840873   &   0.748891  &   0.859860  &     0.379618  \\
100 &     7.336247  &   7.351327  &   7.403872  &  7.244763  &   2.368432   &   2.549805  &   2.314636  &     2.506640  \\
110 &     3.130114  &   3.610598  &   3.107478  &  3.844225  &   5.712811   &   6.211737  &   5.632266  &     7.031858  \\
120 &     0.948158  &   1.548517  &   0.947478  &  2.057045  &  11.081367   &  11.552317  &  11.086289  &    13.170435  \\
      \noalign{\smallskip}\hline\noalign{\smallskip}
    \end{tabular}
\end{table}

\begin{table}
  \centering
  \caption{Computational costs of numerical evaluation of step-down call and put options (for one spot value and five strike prices) under the CEV, CU, UOU, and BK asset price models. The Matlab code was run on a Hewlett-Packard(R) Notebook PC with a four-core Intel(R) Core(TM) i7 CPU Q720 @ 1.6GHz and 4 GB of memory.}
  \label{table1B}
  \setlength{\tabcolsep}{10pt}
    \begin{tabular}{rrrr}
      \hline\noalign{\smallskip}
      Model& $N$ & $T_{\text{sp.exp.}}$ & $T_{\text{quad.}}$ \\
      \hline\noalign{\smallskip}
      CEV  & 150 & 8.81~sec  &  60.70~sec\\
      BK   & 50  & 0.77~sec  &  2.94~sec\\
      CU   & 300 & 20.96~sec & 148.68~sec \\
      UOU  & 150 & 5.36~sec  & 52.63~sec \\
      \noalign{\smallskip}\hline\noalign{\smallskip}
    \end{tabular}
\end{table}

\begin{table}
  \centering
  \caption{The Monte Carlo biased estimates of proportional step-down call and put prices under the CEV model using the Brownian bridge interpolation  method are tabulated for various values of strike $K$.  Monte Carlo estimates are compared with analytical approximations provided by the spectral expansion method. Here, $s_M$ denotes the stochastic error. The parameters used are $S_0 = 100$, $T=1$, $\mathrm{r}=0.1$, $\delta = 2.5$, $\beta = -0.5$, $\alpha = 0.5$, $L=90$, $\Delta t=0.05$. The number of sample paths is $M=10^6$.}
  \label{table2}
  \setlength{\tabcolsep}{10pt}
    \begin{tabular}{rr@{$\,\pm\,$}l@{\;}rr@{$\,\pm\,$}l@{\;}r}
      \hline\noalign{\smallskip}
      &\multicolumn{3}{c}{Step Calls}&\multicolumn{3}{c}{Step Puts}\\
      $K$ & MCM~Estimate & $s_M$   & Analyt.~Approx.    &  MCM~Estimate & $s_M$   &  Analyt.~Approx. \\
      \hline\noalign{\smallskip}
      90  & 20.9950  & 0.0009  & 20.993325  &  2.0416   & 0.0006  &  2.039807\\
      100 & 14.8192  & 0.0006  & 14.817208  &  4.1988   & 0.0010  &  4.195828\\
      110 & 9.8621   & 0.0004  & 9.860234   &  7.5750   & 0.0017  &  7.570991\\
      \noalign{\smallskip}\hline\noalign{\smallskip}
    \end{tabular}
\end{table}

\begin{table}
  \centering
  \caption{Step-down call and put option values are computed under the Bessel-K model for increasing values of~$\alpha$. The option parameters are $K=100$, $T=0.5$, $L=90$. The model parameters are specified in Table~\ref{tb:Models}. The case with $\alpha=\infty$ corresponds to the double knock-out barrier option with barriers $L=90$ and $U=400$. When $\alpha=0$, the step call and put options reduce to the European call and put options, respectively.}
  \label{table3}
  \setlength{\tabcolsep}{10pt}
    \begin{tabular}{rrr}
      \hline\noalign{\smallskip}
      $\alpha$ & \multicolumn{1}{c}{Call Value} & \multicolumn{1}{c}{Put Value} \\
      \hline\noalign{\smallskip}
            0 & 7.525593 & 6.530576\\
            1 & 7.483054 & 5.213809\\
            5 & 7.351327 & 2.549805\\
            10 & 7.240869 & 1.469595\\
            25 & 7.060945 & 0.752496\\
            50 & 6.925459 & 0.524287\\
            100 & 6.806920 & 0.404356\\
            200 & 6.710443 & 0.336245\\
            500 & 6.615400 & 0.285627\\
            1000 & 6.563974 & 0.263218\\
        $\infty$ & 6.494245  & 0.218820  \\
      \noalign{\smallskip}\hline\noalign{\smallskip}
    \end{tabular}
\end{table}

\begin{figure}[ht]
  \centering
\begin{subfigure}[b]{0.475\linewidth}
  \centering
  \includegraphics[width=\linewidth]{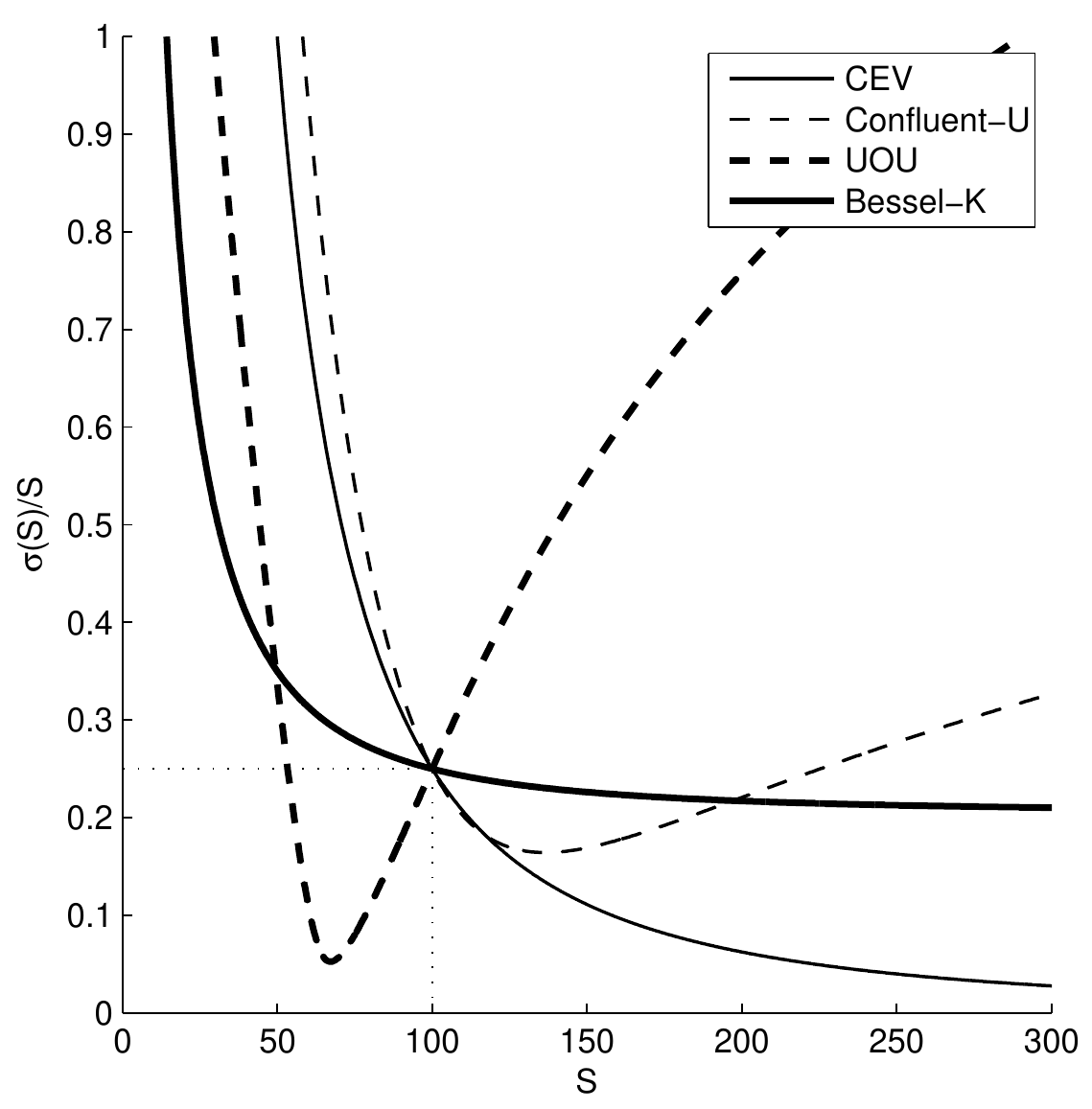}
  \caption{The local volatility $\sigma(S)/S$.}\label{fig:LocVol}
\end{subfigure}
\hfill
\begin{subfigure}[b]{0.475\linewidth}
  \centering
  \includegraphics[width=\linewidth]{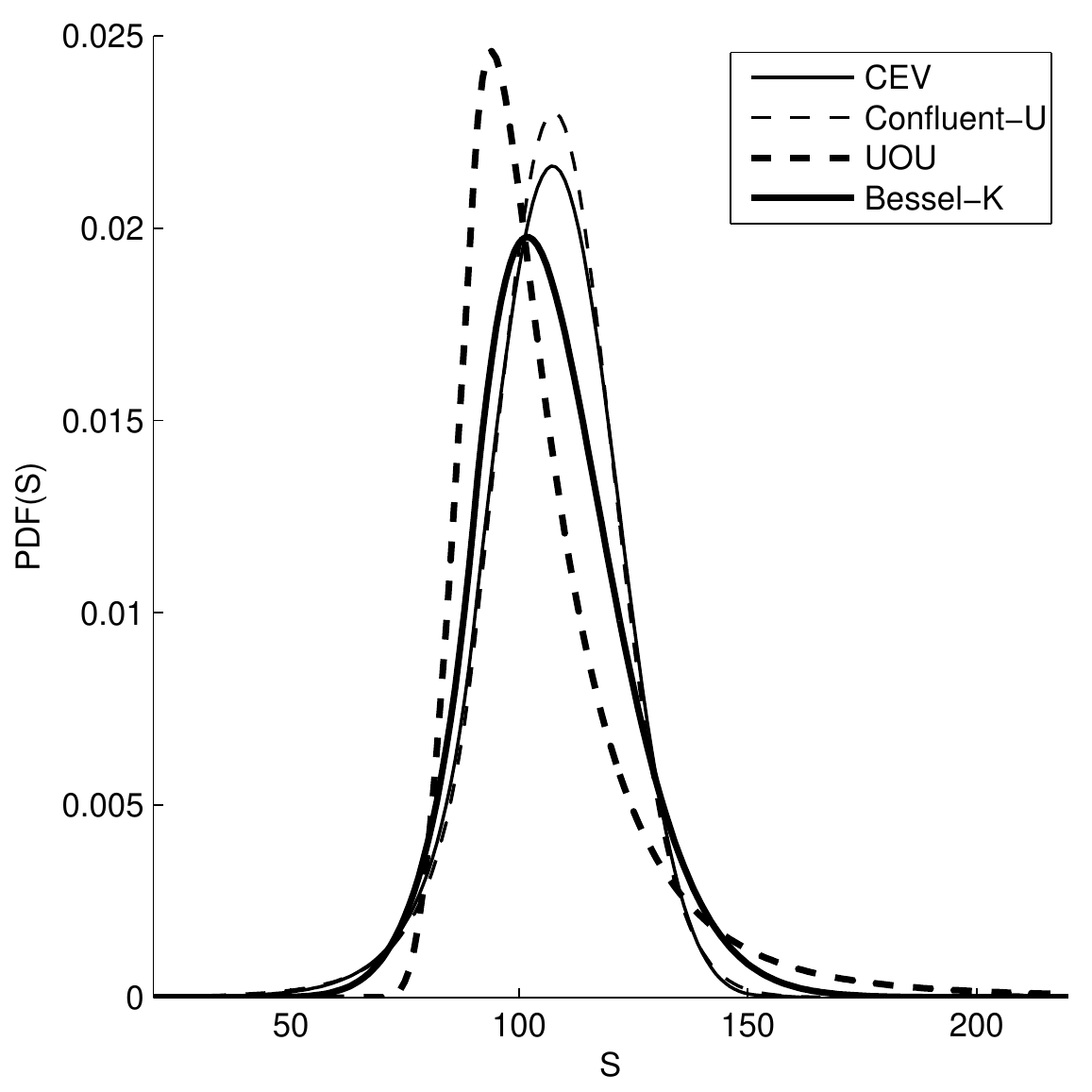}
  \caption{The PDF $\tilde{p}_\alpha^{\ell,-}$.}\label{fig:PDFs}
\end{subfigure}
\caption{The local volatility function $\sigma(S)/S$ and corresponding transition PDFs $\tilde{p}_\alpha^{\ell,-}$, as function of spot $S$, for the asset price process with killing at an exponential stopping time of occupation below a~fixed level $L$, are computed for four asset price models specified in Table~\ref{tb:Models}. The PDFs $\tilde{p}_\alpha^{\ell,-}$ are computed for the following parameters: $S_0=100$, $\alpha=5$, $L=90$, and $T=\frac{1}{2}$.}%
\label{fig:LocVolandPDFs}
\end{figure}

\begin{figure}[ht]
  \centering
\begin{subfigure}[b]{0.475\linewidth}
  \centering
  \includegraphics[width=\linewidth]{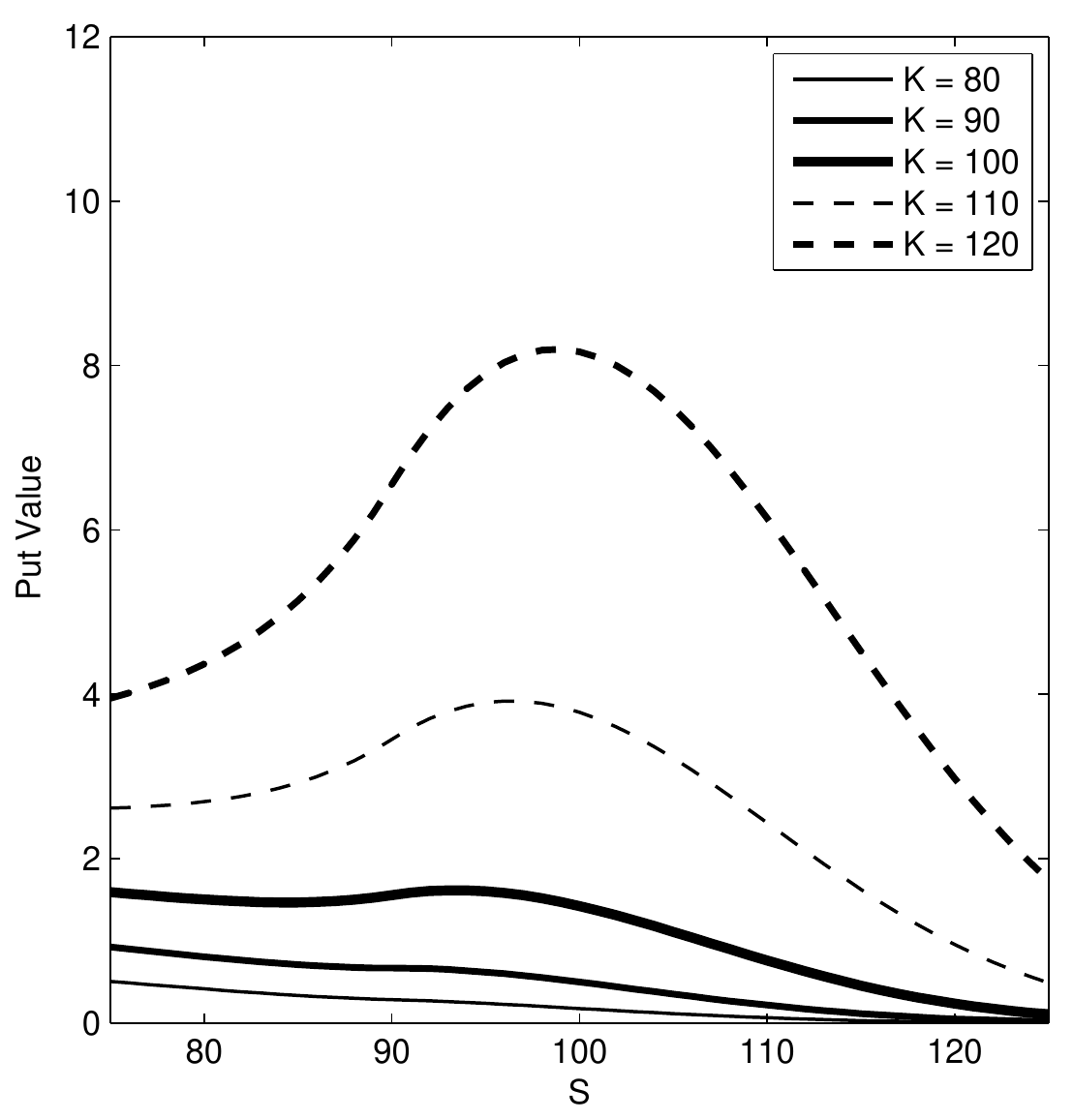}
  \caption{Put values under the CEV model.}\label{fig1a}
\end{subfigure}
\hfill
\begin{subfigure}[b]{0.475\linewidth}
  \centering
  \includegraphics[width=\linewidth]{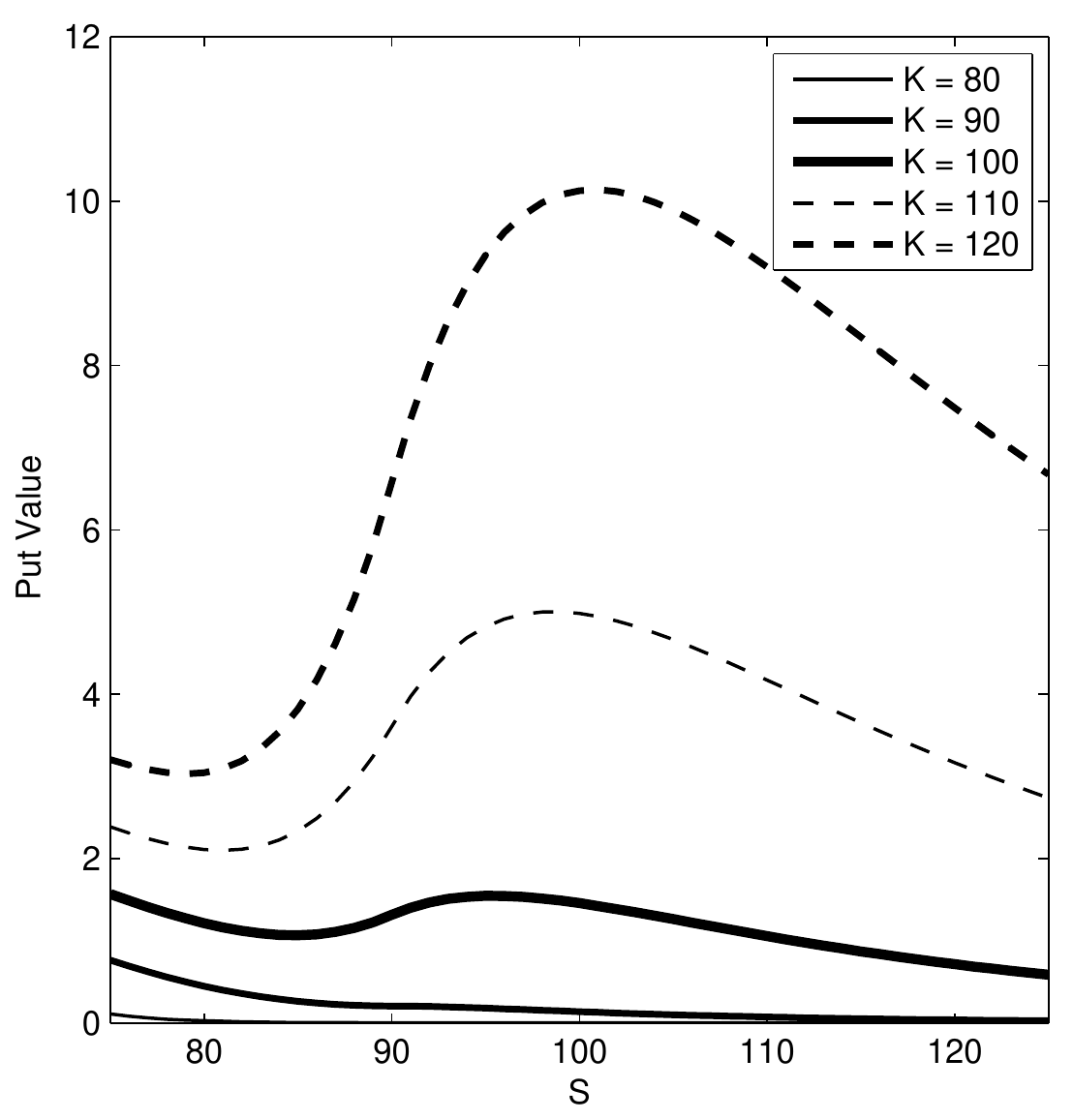}
  \caption{Put values under the UOU model.}\label{fig1b}
\end{subfigure}
\caption{Values of the step-down put option, as function of spot $S$, computed for a~range of strikes under the CEV and UOU models.}%
\label{fig1}
\end{figure}

\begin{figure}[ht]
  \centering
\begin{subfigure}[b]{0.475\linewidth}
  \centering
  \includegraphics[width=\linewidth]{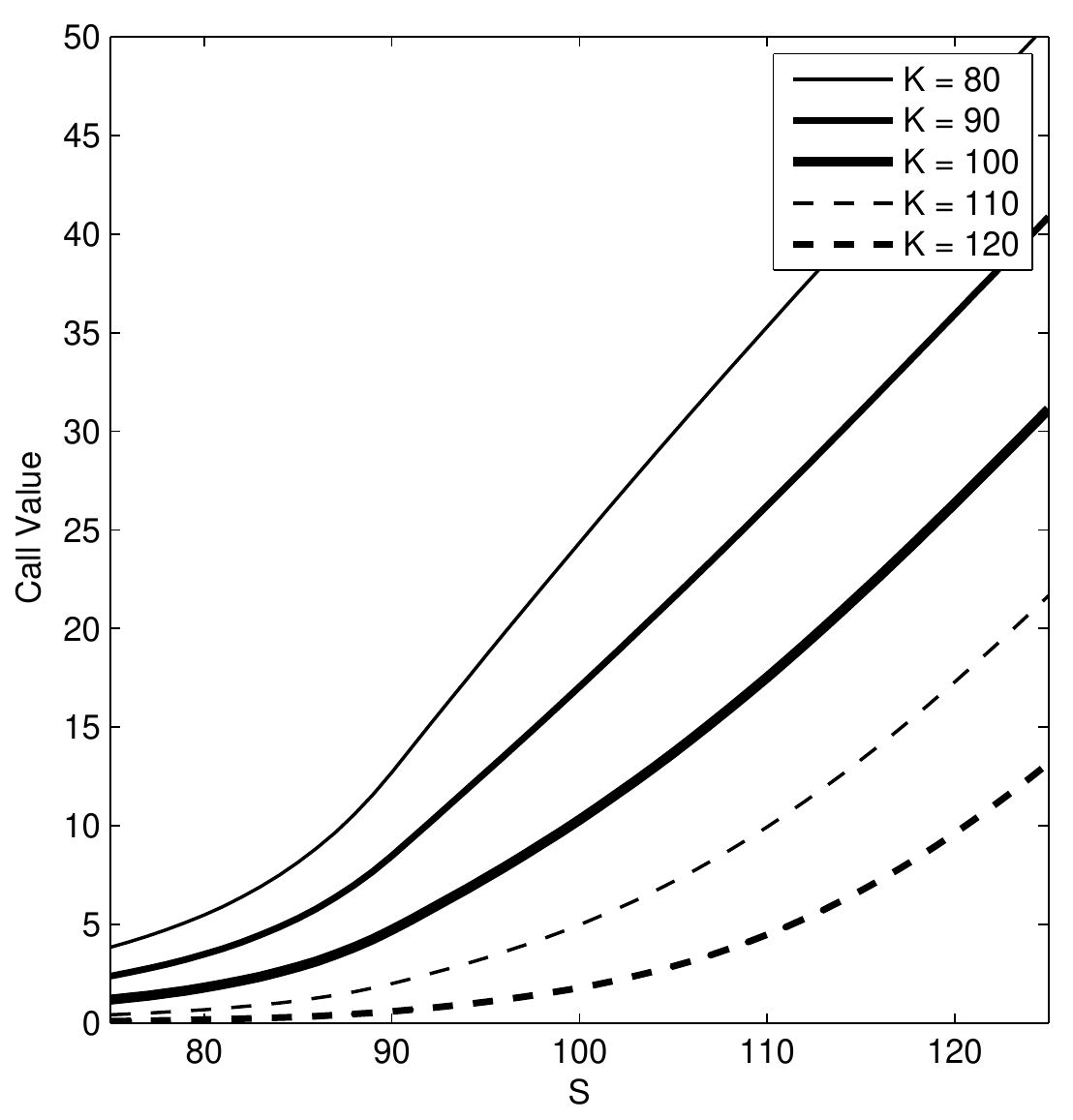}
  \caption{Call values under the CU model.}\label{fig2a}
\end{subfigure}
\hfill
\begin{subfigure}[b]{0.475\linewidth}
  \centering
  \includegraphics[width=\linewidth]{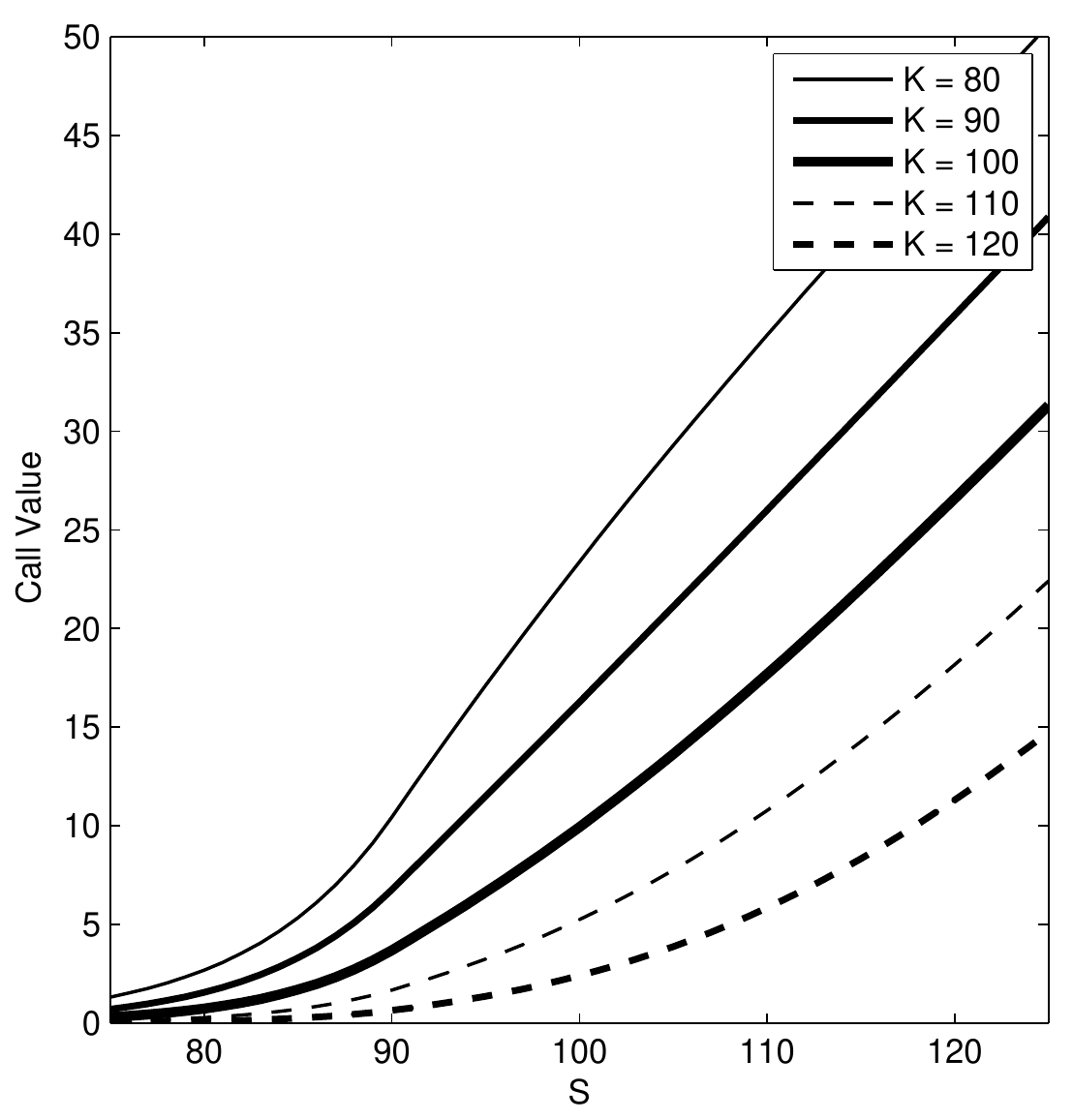}
  \caption{Call values under the BK model.}\label{fig2b}
\end{subfigure}
\caption{Values of the step-down call option, as function of spot $S$, computed for a~range of strikes under the confluent-U and Bessel-K models.}%
\label{fig2}
\end{figure}

\begin{figure}[ht]
  \centering
\begin{subfigure}[b]{0.475\linewidth}
  \centering
  \includegraphics[width=\linewidth]{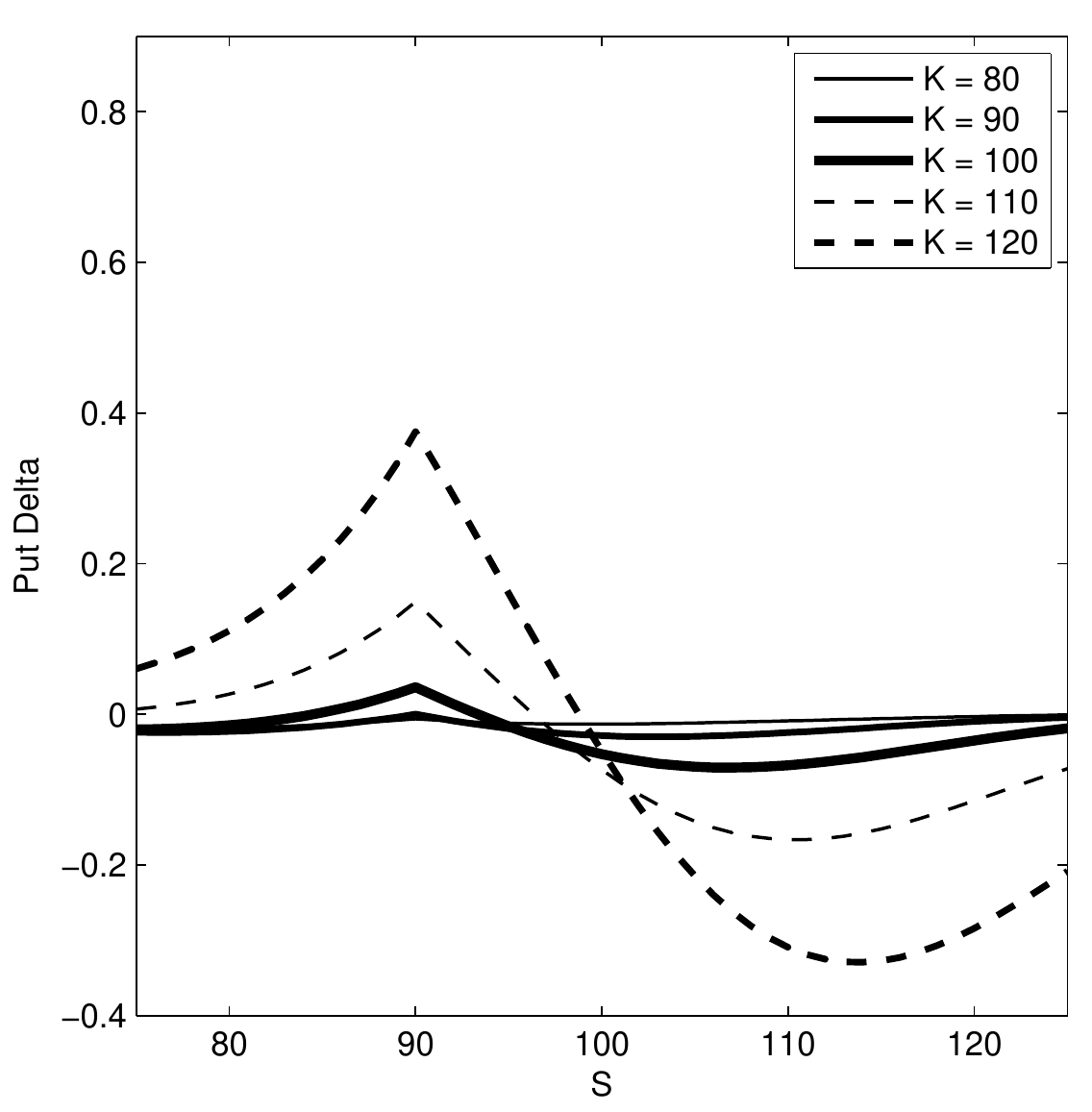}
  \caption{Put deltas under the CEV model.}\label{fig3a}
\end{subfigure}
\hfill
\begin{subfigure}[b]{0.475\linewidth}
  \centering
  \includegraphics[width=\linewidth]{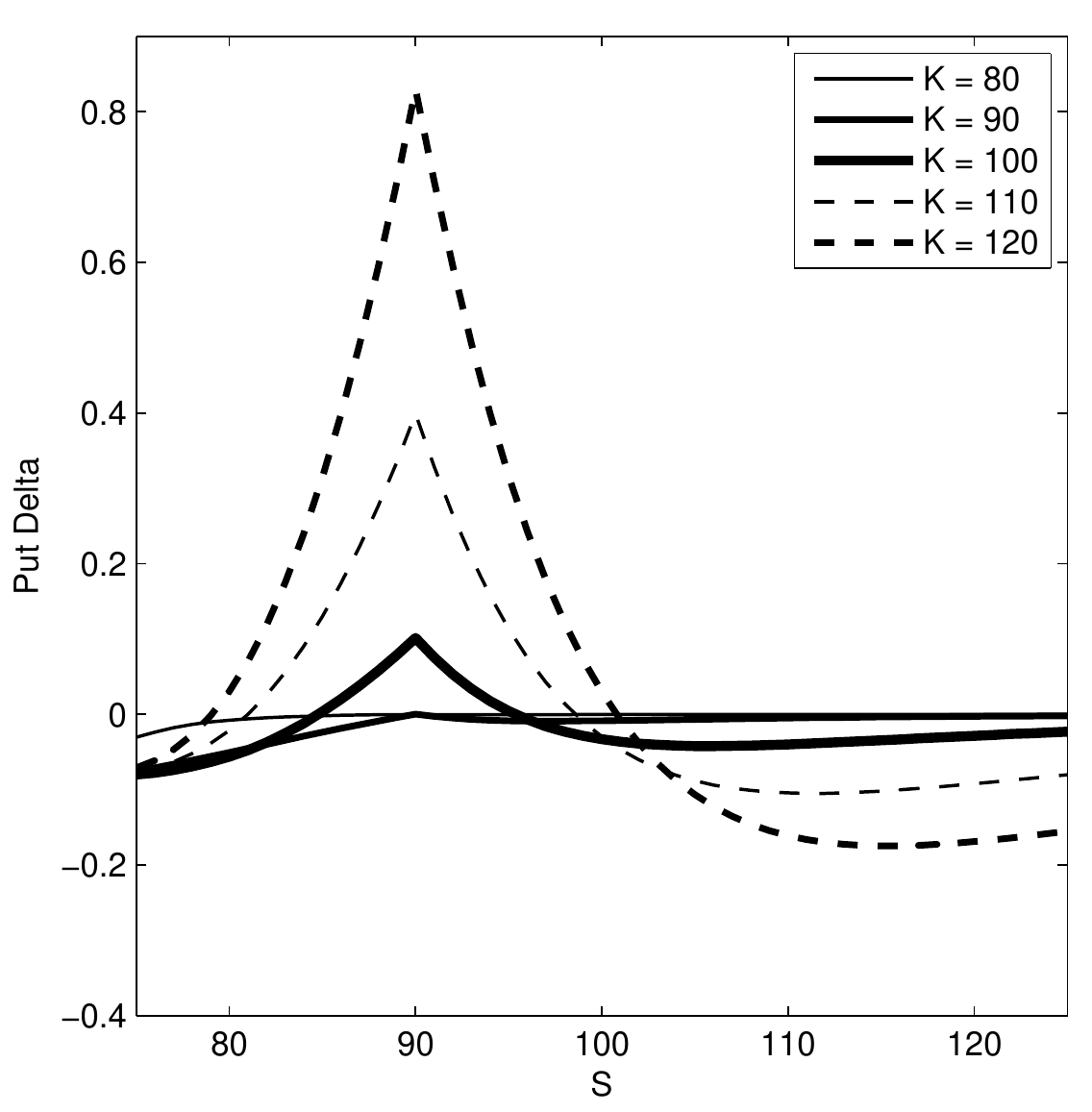}
  \caption{Put deltas under the UOU model.}\label{fig3b}
\end{subfigure}
\caption{Deltas of the step-down put option, as function of spot $S$, computed for a~range of strikes under the CEV and UOU models.}%
\label{fig3}
\end{figure}

\begin{figure}[ht]
  \centering
\begin{subfigure}[b]{0.475\linewidth}
  \centering
  \includegraphics[width=\linewidth]{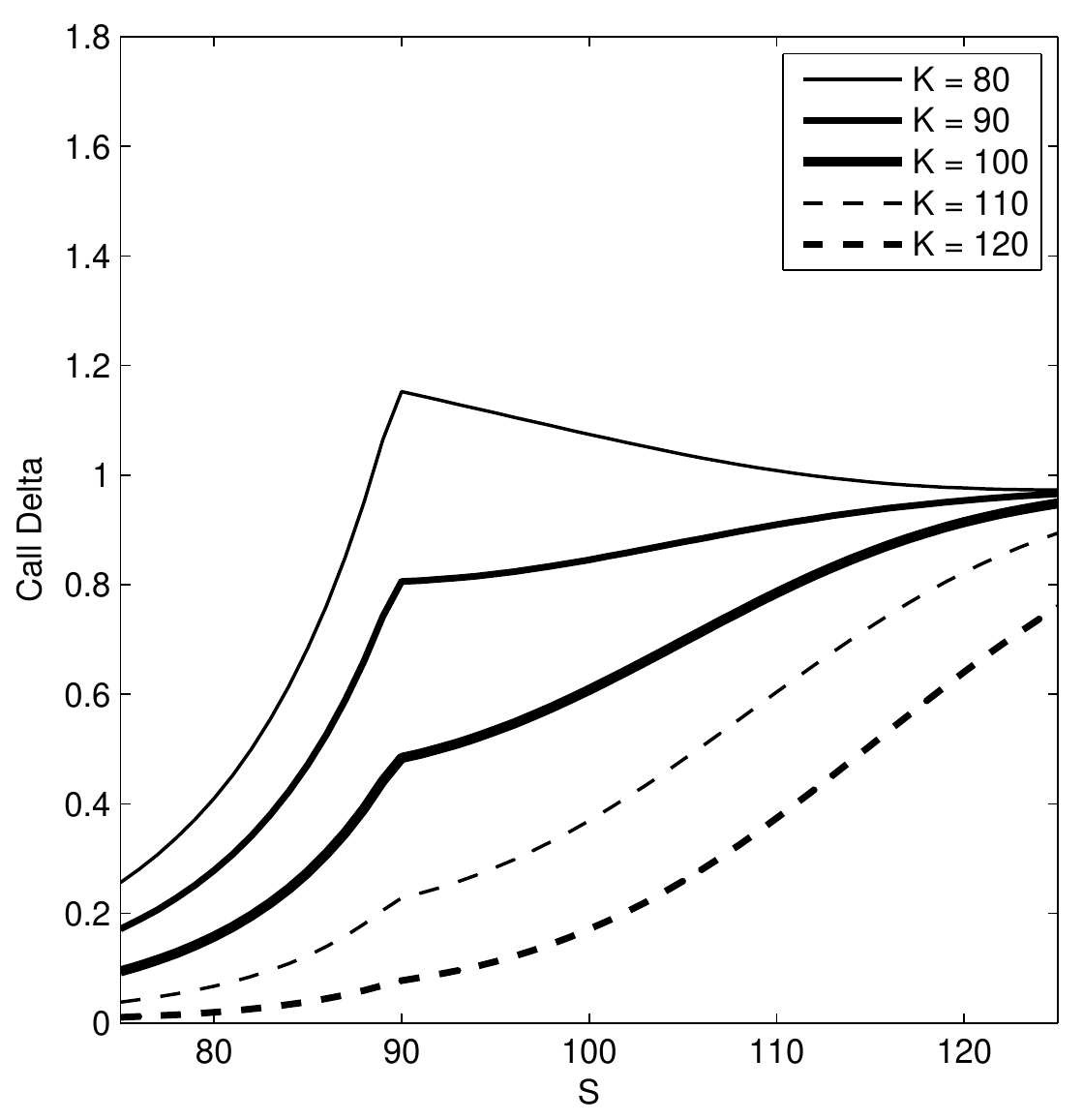}
  \caption{Call deltas under the CU model.}\label{fig4a}
\end{subfigure}
\hfill
\begin{subfigure}[b]{0.475\linewidth}
  \centering
  \includegraphics[width=\linewidth]{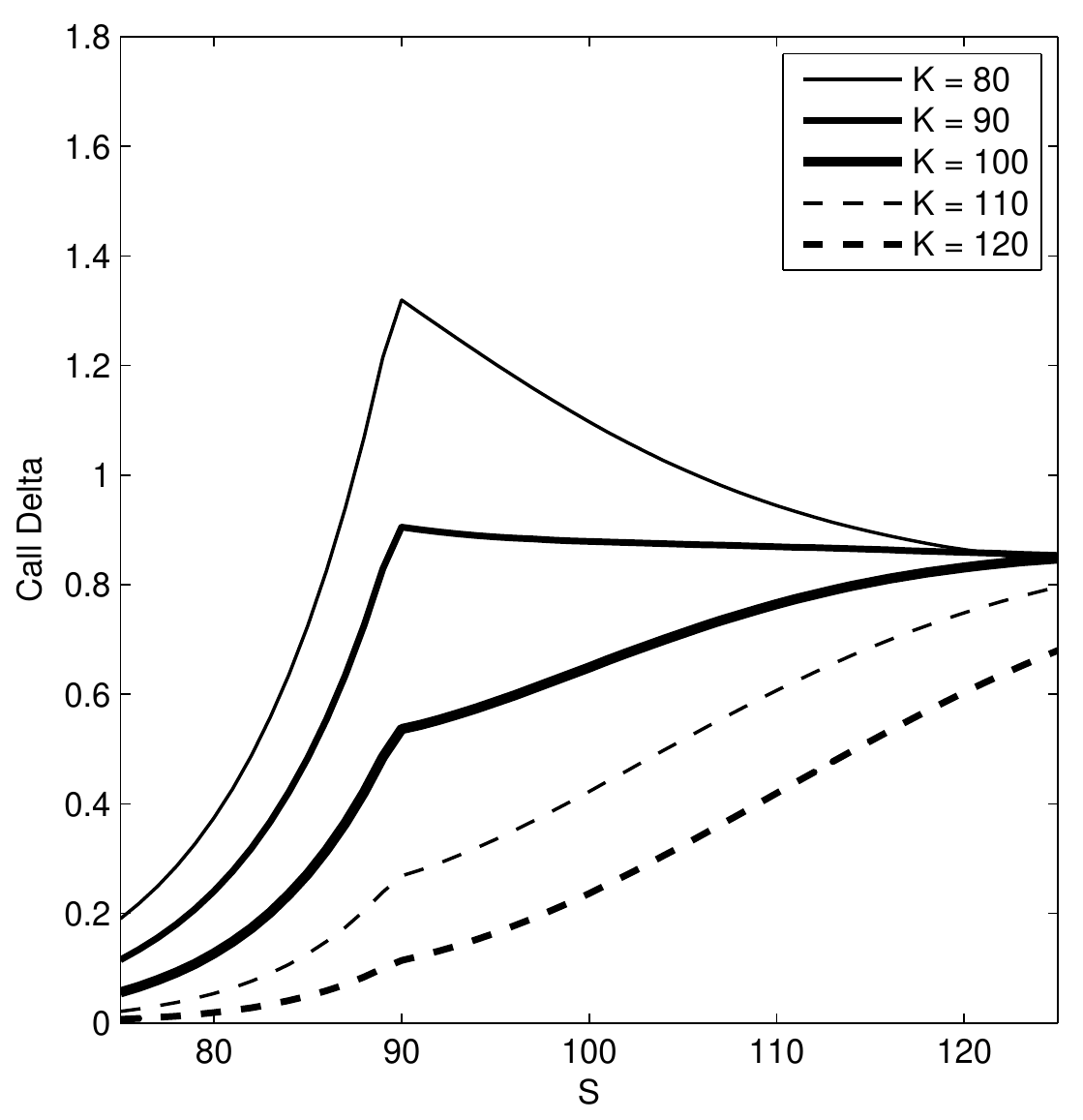}
  \caption{Call deltas under the BK model.}\label{fig4b}
\end{subfigure}
\caption{Deltas of the step-down call option, as function of spot $S$, computed for a~range of strikes under the (a) confluent-U and (b) Bessel-K models.}%
\label{fig4}
\end{figure}

\begin{figure}[ht]
  \centering
\begin{subfigure}[b]{0.475\linewidth}
  \centering
  \includegraphics[width=\linewidth]{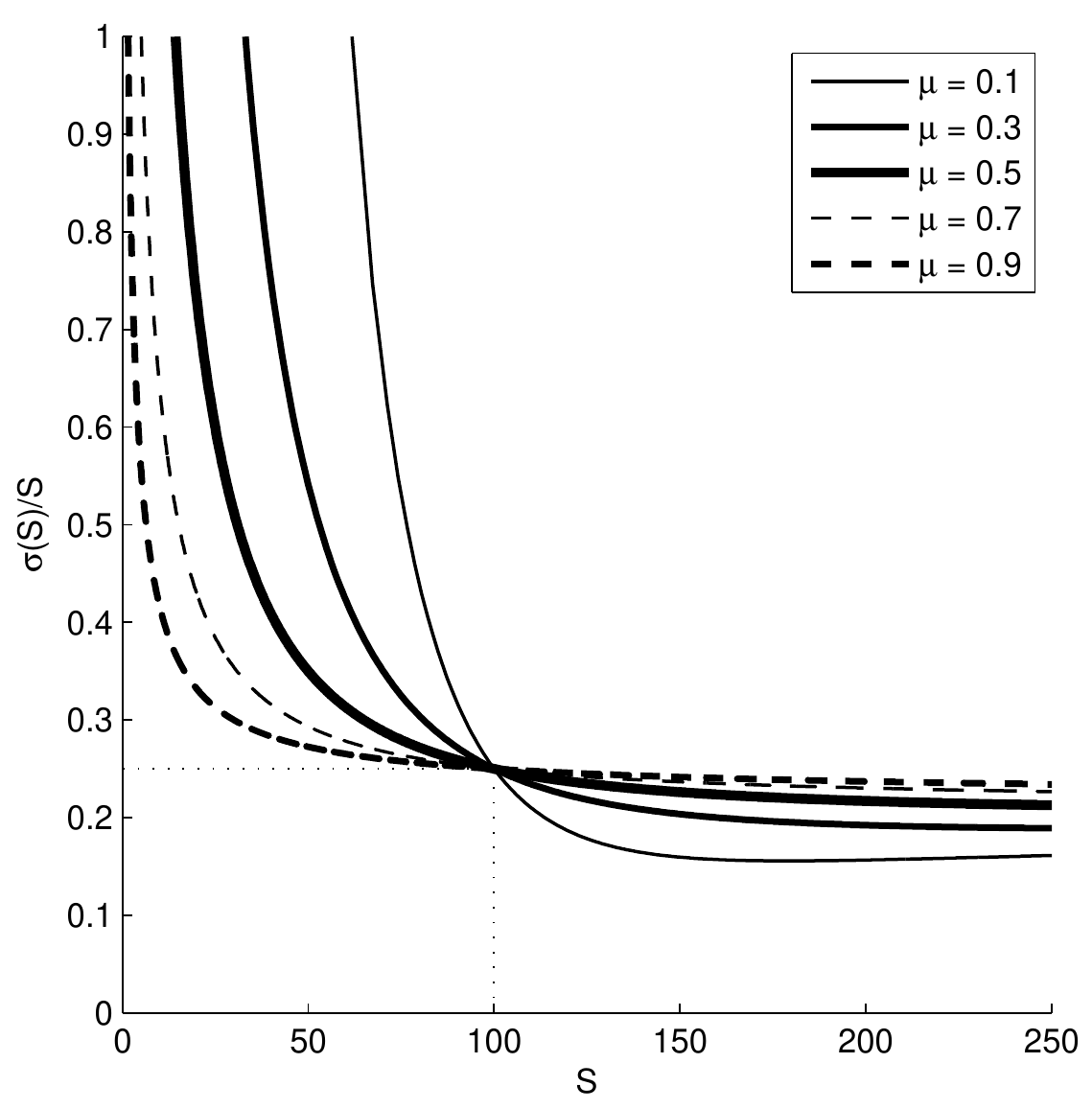}
  \caption{Local volatility functions.}\label{fig:BKmuA}
\end{subfigure}
\hfill
\begin{subfigure}[b]{0.475\linewidth}
  \centering
  \includegraphics[width=\linewidth]{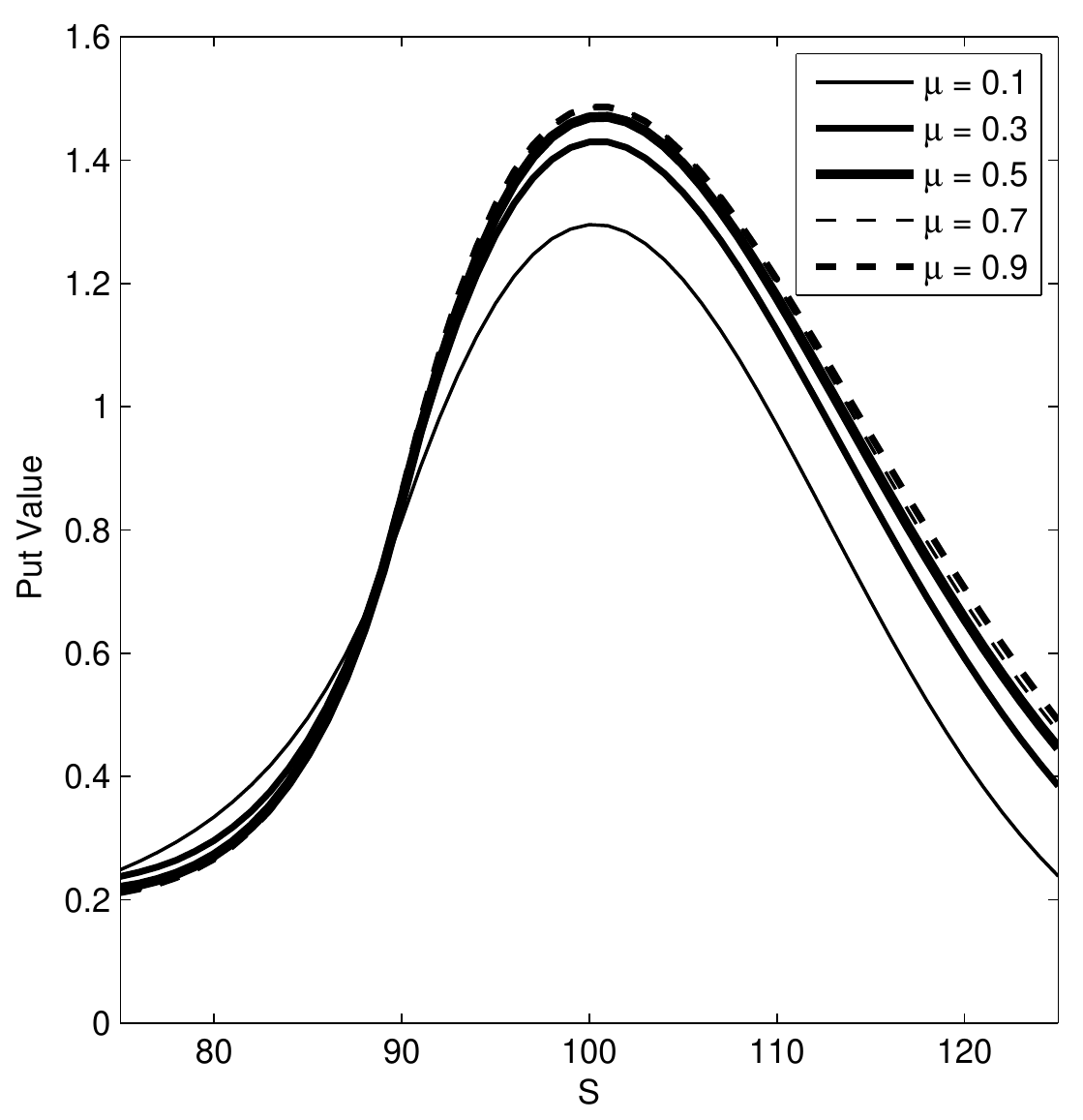}
  \caption{Step put values.}\label{fig:BKmuB}
\end{subfigure}
\caption{Local volatility functions and put option prices computed, as function of spot $S$, under the Bessel-K model for a~range of values of $\mu$. The parameter $c$ varies accordingly so that $\sigma(100)/100=0.25$ for all choices of $\mu$.
The other model parameters are as specified in Table~\ref{tb:Models}. The other parameters are $T=\frac{1}{2}$, $\alpha=5$, $L=90$, and $K=S_0=100$.}
\label{fig:BKmu}
\end{figure}

\begin{figure}[ht]
  \centering
\begin{subfigure}[b]{0.475\linewidth}
  \centering
  \includegraphics[width=\linewidth]{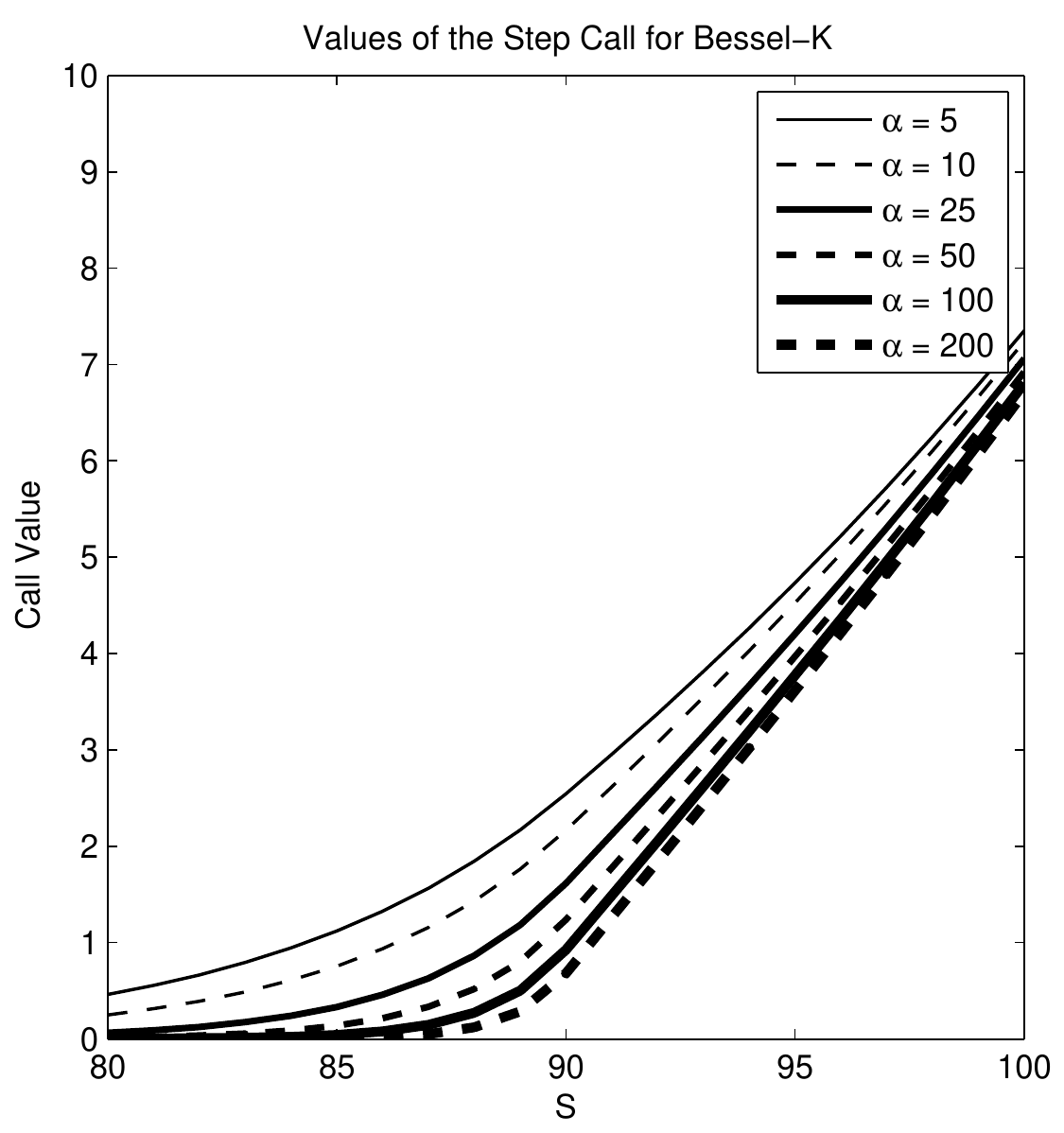}
  \caption{Call values.}\label{fig:BKrhoA}
\end{subfigure}
\hfill
\begin{subfigure}[b]{0.475\linewidth}
  \centering
  \includegraphics[width=\linewidth]{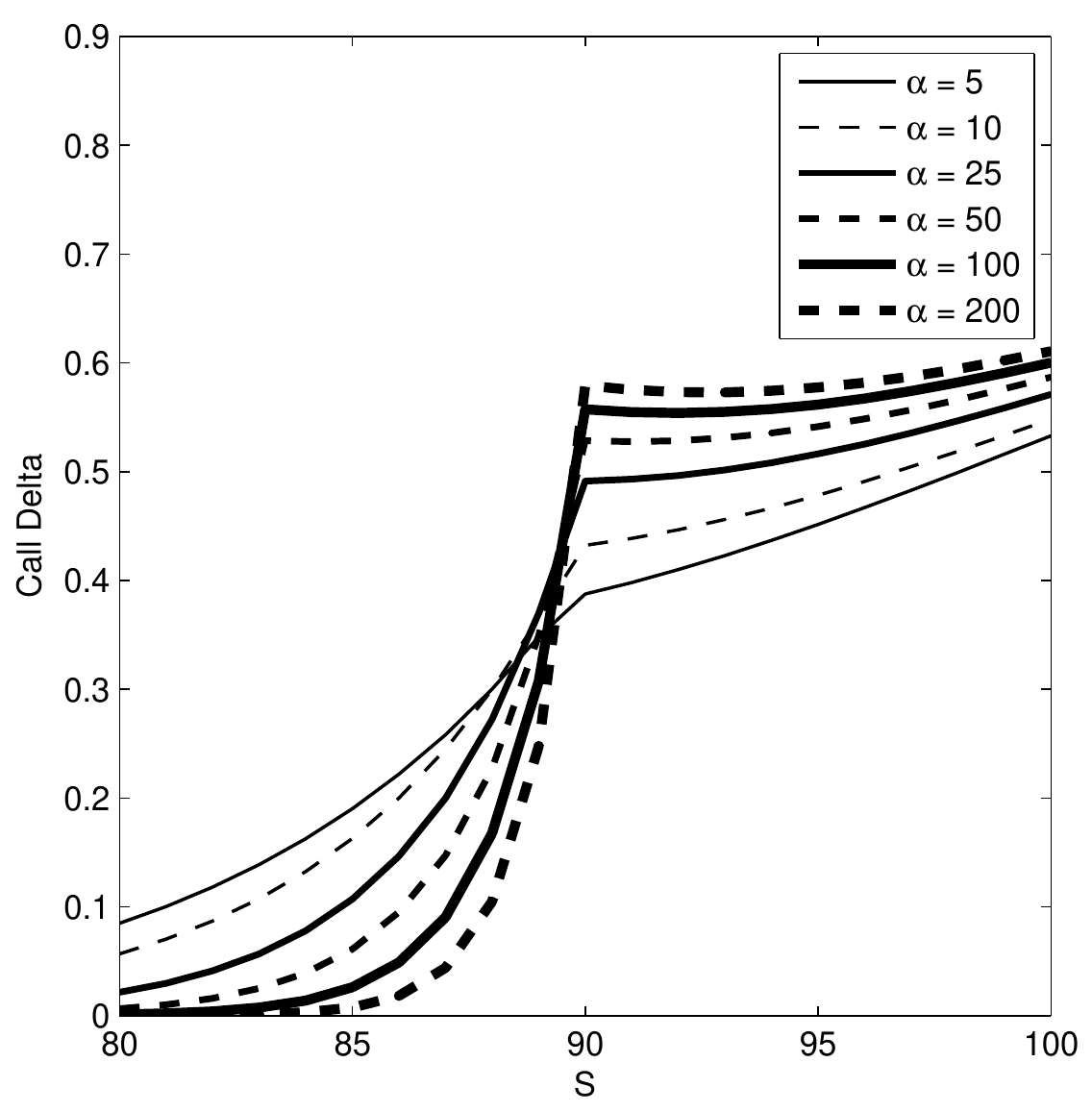}
  \caption{Call deltas.}\label{fig:BKrhoB}
\end{subfigure}
\caption{Step-down call option values and deltas are computed under the Bessel-K model (as specified in Table~\ref{tb:Models}) for increasing values of~$\alpha$. The option parameters are $K=100$, $T=0.5$, $L=90$.}%
\label{fig:BKrho}
\end{figure}

\begin{figure}[ht]
  \centering
\begin{subfigure}[b]{0.475\linewidth}
  \centering
  \includegraphics[width=\linewidth]{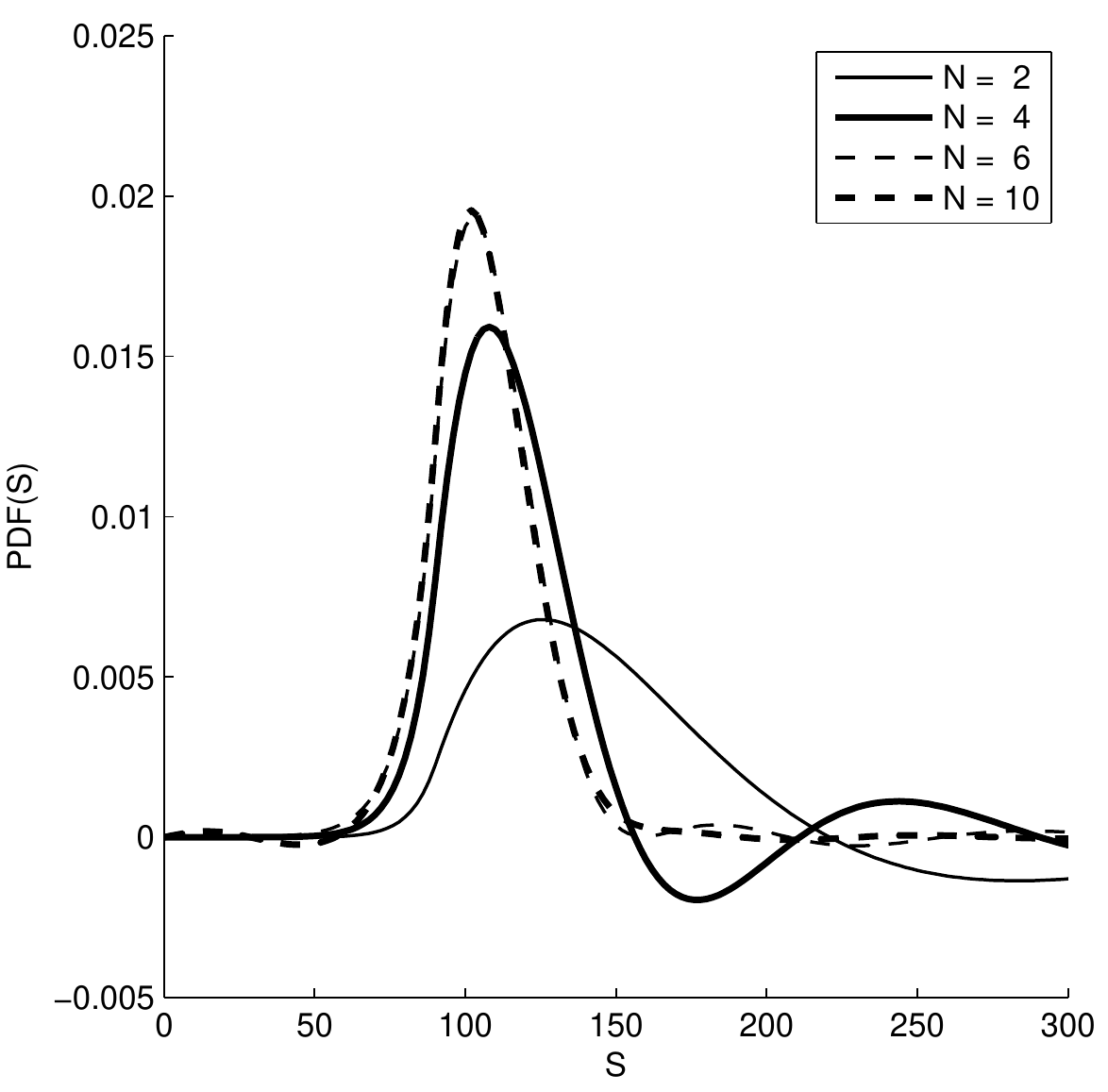}
  \caption{Approximations of the PDF $\tilde{p}_\alpha^{\ell,-}$.}\label{fig:convA}
\end{subfigure}
\hfill
\begin{subfigure}[b]{0.5\linewidth}
  \centering
    \begin{tabular}{rrr}
      \hline\noalign{\smallskip}
      $N$ & Call Value & Put Value \\
      \hline\noalign{\smallskip}
       4 & 0.818540 & 1.705722\\
       6 & 6.290608 & 2.867394\\
       10 & 7.593736 & 2.745855\\
       15 & 7.377383 & 2.539072\\
       20 & 7.351545 & 2.550244\\
       50 & 7.351327 & 2.549805\\
      \noalign{\smallskip}\hline\noalign{\smallskip}
    \end{tabular}
  \vspace{1in}
  \caption{Call and put values.}\label{fig:convB}
\end{subfigure}
\caption{The convergence of truncated series approximations of the PDF $\tilde{p}_\alpha^{\ell,-}$, as the number of terms $N$ increases. The computations were done for the Bessel-K model whose parameters are specified in Table~\ref{tb:Models}. The other parameters are $T=\frac{1}{2}$, $\alpha=5$, $L=90$, and $S_0=100$. The step-down option values are computed for the strike price $K=100$.}%
\label{fig:conv}
\end{figure}

\begin{figure}[ht]
  \centering
\begin{subfigure}[b]{0.475\linewidth}
  \centering
  \includegraphics[width=\linewidth]{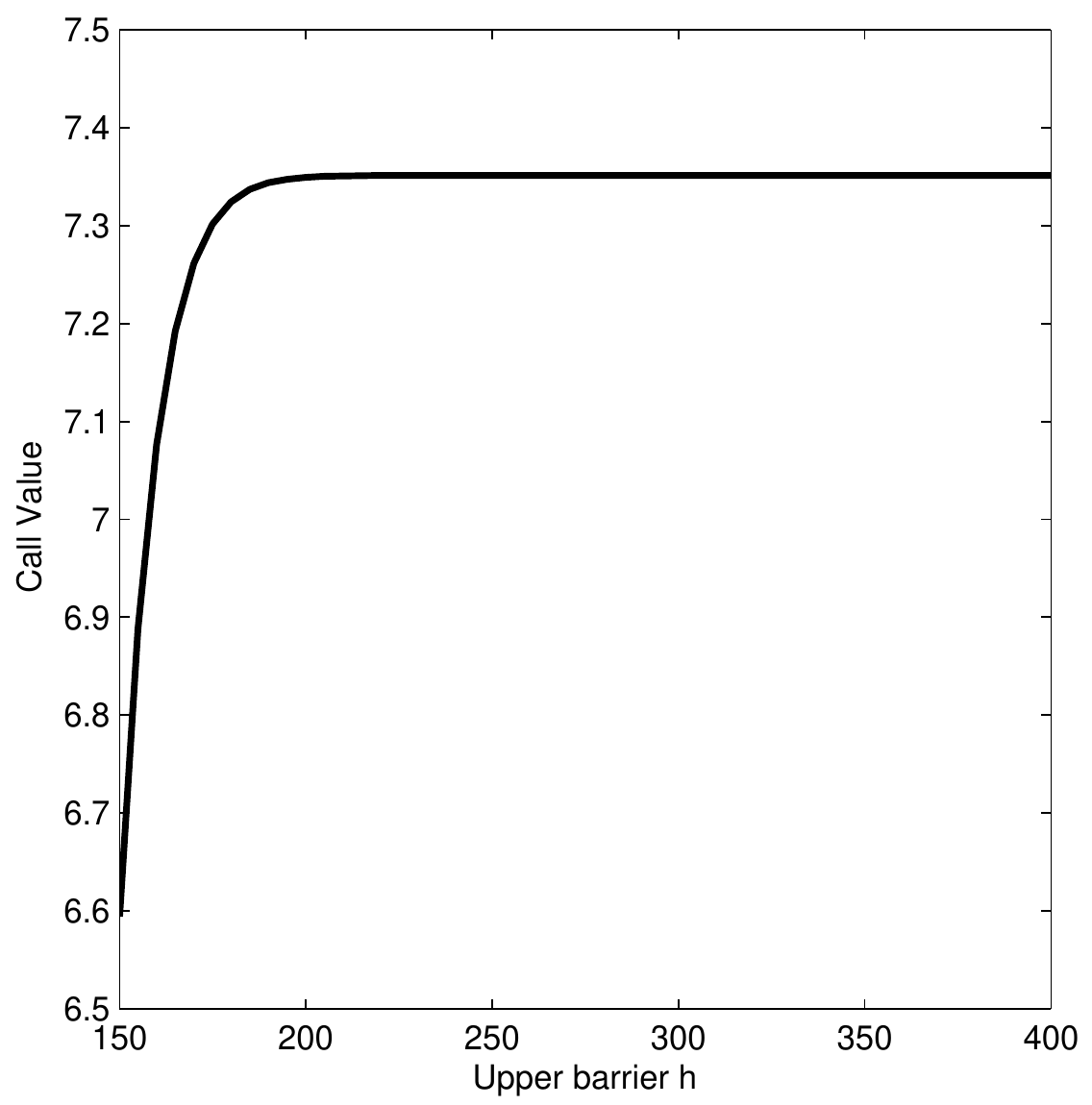}
  \caption{The call value as function of $\h $.}\label{fig:levelA}
\end{subfigure}
\hfill
\begin{subfigure}[b]{0.5\linewidth}
  \centering
    \begin{tabular}{rrr}
      \hline\noalign{\smallskip}
      $\h $ & Call Value & Put Value \\
      \hline\noalign{\smallskip}
      150 & 6.594098 & 2.549803\\
      160 & 7.076232 & 2.549805\\
      170 & 7.261562 & 2.549805\\
      180 & 7.324466 & 2.549805\\
      190 & 7.343837 & 2.549805\\
      200 & 7.349354 & 2.549805\\
      225 & 7.351269 & 2.549805\\
      250 & 7.351326 & 2.549805\\
      275 & 7.351327 & 2.549805\\
      300 & 7.351327 & 2.549805\\
      350 & 7.351327 & 2.549805\\
      400 & 7.351327 & 2.549805\\
      \noalign{\smallskip}\hline\noalign{\smallskip}
    \end{tabular}
  \vspace{0.5in}
  \caption{Call and put prices.}\label{fig:levelB}
\end{subfigure}
\caption{Convergence of the prices $C^-_{\mathrm{step}}(S_0=100,T=0.5,K=100)$ of the step-down call option as the imposed killing level $\h $ increases. The computations are for the Bessel-K model whose parameters are specified in Table~\ref{tb:Models}. The other parameters are $\alpha=5$ and $L=90$.}%
\label{fig:level}
\end{figure}

\end{document}